\documentclass[aap,preprint]{imsart}

\RequirePackage[OT1]{fontenc}
\RequirePackage{amsthm,amsmath}
\RequirePackage[numbers]{natbib}
\RequirePackage[colorlinks,citecolor=blue,urlcolor=blue]{hyperref}

\arxiv{1408.1382}

\startlocaldefs
\usepackage{amsmath, amsfonts, amssymb,  mathrsfs,  array, stmaryrd, fullpage, indentfirst,  amsthm,  hyperref}

\numberwithin{equation}{section}
\theoremstyle{plain}
\newtheorem{lemma}{Lemma}[section]
\newtheorem{example}{Example}[section]
\newtheorem{theorem}{Theorem}[section]
\newtheorem{definition}{Definition}[section]
\newtheorem{proposition}{Proposition}[section]
\newtheorem{corollary}{Corollary}[section]
\newtheorem{assumption}{Assumption}[section]
\theoremstyle{remark}
\newtheorem{remark}{Remark}
\newcommand{\ab}[1]{\langle #1\rangle}

\newcommand{\cadlag}{c\`{a}dl\`{a}g }

\endlocaldefs

\begin{document}

\begin{frontmatter}
\title{Optimal Consumption under Habit Formation In Markets with Transaction Costs and Random Endowments}
\runtitle{HABIT FORMATION AND TRANSACTION COSTS}

\begin{aug}
\author{\snm{Xiang Yu}  \ead[label=e1]{xiang.yu@polyu.edu.hk}},

\runauthor{X.YU}

\affiliation{The Hong Kong Polytechnic University}

\address{Department of Applied Mathematics, Yip Kit Chuen Building\\
The Hong Kong Polytechnic University\\
Hung Hom, Kowloon, Hong Kong\\
E-mail:\ \printead*{e1}}

\end{aug}

\begin{abstract}
This paper studies the optimal consumption via the habit formation preference in markets with transaction costs and unbounded random endowments. To model the proportional transaction costs, we adopt the Kabanov's multi-asset framework with a cash account. At the terminal time $T$, the investor can receive unbounded random endowments for which we propose a new definition of acceptable portfolios based on the strictly consistent price system (SCPS). We prove a type of super-hedging theorem using the acceptable portfolios which enables us to obtain the consumption budget constraint condition under market frictions. Following similar ideas in \cite{Yu1} with the path dependence reduction and the embedding approach, we obtain the existence and uniqueness of the optimal consumption using some auxiliary processes and the duality analysis. As an application of the duality theory, the market isomorphism with special discounting factors is also discussed in the sense that the original optimal consumption with habit formation is equivalent to the standard optimal consumption problem without the habits impact, however, in a modified isomorphic market model.
\end{abstract}



\end{frontmatter}

\ \\
\section{Introduction}

The study of consumption habit formation in financial economics dates back to \cite{Hicks} and \cite{Heal}. It has been observed that the von Neumann-Morgenstern utilities can not reconcile the well-known magnitude of the equity premium (See \cite{Mehra} and \cite{constantinides1988habit}). Instead, the habit formation preference takes care of both the current consumption choice and its history pattern. Due to its time non-separable structure, a small change in the consumption can cause a large fluctuation in consumption net of the subsistence level, which may explain the sizable excess returns on risky assets in equilibrium models. Likewise this new preference can shed some light on the quantitative explanation of consumer's psychology in the context of the consumption behavior with investment opportunities. For instance, many empirical studies and psychological reports reveal that the consumer's satisfaction level and risk tolerance sometimes rely more on recent changes than the absolute levels. This preference has thereby been a surge in models for which the smooth consumption is more beneficial than the marked increase, such as the household consumption and other types of expenditures with commitment. In addition, as pointed out in \cite{Fuhrer}, habit formation allows the model to match the response of real spending to monetary policy shocks and this new preference can also more accurately replicate the gradual decline of inflation during a disinflation. The diverse applications of this path-dependent preference motivate our research in general incomplete market models.

Mathematically, the habit formation preference is defined by $\mathbb{E}[\int_0^TU(t,c_t-F(c)_t)dt]$, where $U:[0,T]\times(0,\infty)\rightarrow\mathbb{R}$. The accumulative process $(F(c)_t)_{t\in[0,T]}$, called the \textit{habit formation} or \textit{the standard of living} process, describes the consumption history impact. The conventional definition (see \cite{detemple91} and \cite{detemple92}) of $F(c)_t$ is given by the recursive equation
\begin{equation}
\begin{split}
dF(c)_t&=(\delta_tc_t-\alpha_tF(c)_t)dt,\\
F(c)_0&=z,\nonumber
\end{split}
\end{equation}
where the discounting factors $(\alpha_t)_{t\in[0,T]}$ and $(\delta_t)_{t\in[0,T]}$ are assumed to be non-negative optional processes and the given real number $z\geq 0$ is called the \textit{initial habit}. In the present paper, the consumption habits are assumed to be \textit{addictive} in the sense that $c_t\geq F(c)_t$ for all $t\in[0,T]$. Therefore, the consumption rate shall never fall below the standard of living level. This addictive path-dependent preference has been proposed as the new paradigm in the literature and extensively studied over the past few decades, see among \cite{Campbell}, \cite{detemple91}, \cite{detemple92}, \cite{Schroder01072002}, \cite{Eng09}, \cite{adRoman} and \cite{Yu1}.

In frictionless incomplete semimartingale markets, \cite{Yu1} recently solved this optimization problem using the convex duality approach. In \cite{Yu1}, the complexity caused by the path dependence can be reduced by working on the auxiliary primal processes as well as the auxiliary dual processes. However, the stochastic factors $(\alpha_t)_{t\in[0,T]}$ and $(\delta_t)_{t\in[0,T]}$ appeared as some shadow random endowments in the formulation of the auxiliary optimization problem. In addition, in virtue of the case when $\alpha$ and $\delta$ are unbounded, the auxiliary dual process can not be guaranteed to be integrable. By making the asymptotic growth assumption of the utility function at both $x\rightarrow 0$ and $x\rightarrow \infty$, \cite{Yu1} managed to modify the proofs of the duality theory in \cite{kram04} in a delicate way to deal with the shadow random endowment. By the idea of embedding, the existence and uniqueness of the optimal consumption under habit formation are consequent on the duality theory for the auxiliary time-separable problem.

In the presence of transaction costs, the existence of the optimal consumption with habit formation in general market models, however, is still an open problem. From a financial point of view, it is natural to investigate the consumption streams under the influence of habit formation constraints as well as the trading frictions. Intuitively speaking, both the preservation of the standard of living and transaction costs will potentially entail negative impacts on the trading frequency and suggest to allocate more wealth to retain the smooth consumption. However, we expect that these effects are implicit and complicated. To begin with, this paper aims to investigate this open problem by building the theoretical foundation of the existence and uniqueness of the optimal solution. In some concrete models, the study of the interplay of the habit formation and transaction costs as well as the sensitivity analysis of the consumption streams will be left as future projects.

In this paper, we shall follow Kabanov's multiple assets framework with a cash account in which the proportional transaction costs are modeled via a nonnegative matrix. The investor can choose the intermediate consumption from the cash account and will receive unbounded random payoffs at the terminal time $T$ from some contingent claims. It is worth noting that the Bipolar relationship plays an important role in the duality approach. In frictionless markets (see \cite{Yu1}), the Bipolar result relies on the consumption budget constraint which is taken as granted due to the optional decomposition theorem, see \cite{Kramop}. This inequality characterization of the consumption processes suggests us to define the correct auxiliary set and the auxiliary dual set in \cite{Yu1}. Unfortunately, the optional decomposition theorem is no longer valid in our framework since the semimartingale property and the stochastic integral theory are missing in general models with transaction costs. In addition, the conventional admissible portfolios are defined carefully based on the convex solvency cones and the strictly consistent price system (SCPS), see \cite{MAFI:MAFI004}, \cite{Kab} and \cite{Campi06} and the num\'{e}raire-based version by \cite{Pa}. In the existing literature, it is required that each admissible portfolio process is bounded from below by a constant. Under different conditions, some super-hedging theorems can be obtained using the corresponding admissible strategies. This definition of admissible portfolios becomes inappropriate when unbounded random endowments are taken into account. In frictionless markets, \cite{kram04} and \cite{Zit05} used the acceptable portfolio with a process lower bound and proved a type of super-hedging theorem for some workable contingent claims. The key ingredient in this definition is that the set of equivalent local martingale measures (ELMM) such that the given maximal element in the set of wealth processes is a uniformly integrable martingale is dense in the set of all ELMM with respect to the norm topology of $\mathbb{L}^1(\Omega,\mathcal{F},\mathbb{P})$. In our framework, one natural way to modify this definition is to consider the maximal element in the set of admissible portfolios for transaction costs as the lower bound. Nevertheless, this result fails in general as the maximal element is not a uniformly integrable martingale under any SCPS, see some counterexamples in \cite{jacka2007} in the discrete time setting. Since the definition of working portfolios is even more complicated in continuous time models, it is reasonable to believe that the maximal element from the admissible portfolios is not a wise choice. On the other hand, if we choose an arbitrary process as the lower bound, the primal set in the duality theory is not necessarily closed. Evidently, unbounded random final payoffs prohibit us to apply the well established result in the literature, i.e., the super-hedging theorem using admissible portfolios. To summarize the new challenges in the current work, the issues of an appropriate definition of working portfolios and the consumption budget constraint using new portfolios need to be addressed.

Our first contribution to the existing literature is to propose an innovative definition of acceptable portfolio processes. It is important to note that each SCPS can be equivalently written as a pair of $(\mathbb{Q}, \tilde{S})$ where $\tilde{S}$ is a local martingale under $\mathbb{Q}$. Despite the fact that the stock price process $S$ may not be a semimartingale, each $\tilde{S}$ is. Hence, we can define the maximal element from the stochastic integrals using $\tilde{S}$ from SCPS, and apply it as the process lower bound in a proper way for the self-financing portfolios with transaction costs. This definition is expected to be more complicated, however, we provide some handy criterions to check whether the portfolio process is acceptable or not. It is not surprising that the most challenging work in this paper is the proof of the super-hedging theorem. Comparing with \cite{Kramop}, it has been pointed out that we are lack of semimartingale properties and the optional decomposition theorem. On the other hand, our result differs from \cite{Campi06} since a more complicated definition of working portfolios is required in this paper. Fortunately, under our careful choice of acceptable portfolios, the super-hedging result can still be verified. Without consumption behavior, the super-hedging theorem is interesting for its own sake. Using this result, we can also work on the multi-variate utility maximization problem defined on multiple stocks with unbounded random endowments, see \cite{campi1}. In Assumption $3.2$ and Theorem $3.2$ of \cite{campi1}, we can relax the condition that the random endowment $\mathcal{E}\in\mathbb{L}^{\infty}$.

Another main contribution of the current work is to build the rigorous duality theory in the model with transaction costs, unbounded random endowments and consumption habit formation. Our result extends the well known duality theories in \cite{Campi06} to the scenario with habit formation constraints as well as in \cite{kram04} to the market frictions caused by transaction costs. To begin with, the consumption budget constraint provides an inequality characterization of the set of all financeable consumption strategies. By introducing the auxiliary primal process $(\tilde{c}_t)_{t\in[0,T]}$ where $\tilde{c}_t=c_t-F(c)_t$ and the auxiliary dual process $(\Gamma_t)_{t\in[0,T]}$ defined via the SCPS, we can follow and modify the ideas and proofs in \cite{Yu1} to apply the convex duality approach for an abstract optimization problem. It is worth pointing out that the mathematical approach for habit formation may allow us to consider some related economic problems such as the consumption with durable goods under transaction costs. Our general result is also the first step to examine the equilibrium theorem for habit formation preference with market frictions.

As an interesting observation from the duality theory and consumption budget constraint, this paper also aims to discuss the market isomorphism between the markets with consumption habit formation and the markets without habit memories similar to \cite{Schroder01072002}, however, trading in both markets will incur transaction costs. For some special discounting factors $(\alpha_t)_{t\in[0,T]}$ and $(\delta_t)_{t\in[0,T]}$, our original optimization problem is equivalent to a standard time-separable utility maximization problem on consumption under the change of num\'{e}raire in the isomorphic market with modified random endowments. Two special examples are presented that the external num\'{e}raire process will even vanish in the isomorphic optimal consumption problem. This market isomorphism provides a shortcut to examine the optimal consumption strategy under habit influences in this paper using the results on standard optimal consumption under transaction costs in the existing literature. In particular, if the shadow price process also exists for the given asset price process and transaction costs, the market isomorphism into the shadow price model will enable us to obtain closed-form feedback formula for the optimal consumption choice.

The rest of the paper is organized as follows: Section $\ref{section1}$ introduces the market model with transaction costs and unbounded random endowments, however, without intermediate consumption behavior. We propose the new definition of acceptable portfolio appropriate for market frictions using SCPS. The super-hedging theorem for a family of workable contingent claims is derived. Section $\ref{section2}$ is devoted to the case with intermediate consumption and addictive habit formation. To reduce the path dependence, we define the auxiliary primal space $\bar{\mathcal{A}}(x,q,z)$, the enlarged space $\widetilde{\mathcal{A}}(x,q,z)$ and the auxiliary dual space $\widetilde{\mathcal{M}}$. The original problem is embedded into an abstract time separable optimization problem with some shadow random endowments. In Section $\ref{section3}$, we formulate the auxiliary dual problem in which the random endowments can be hidden. The main theorem is stated in the end together with some corollaries. For special choices of discounting factors, the market isomorphism result is investigated in Section $\ref{section4}$. Section $\ref{section5}$ presents proofs of all main results in previous sections.\\

\ \\

\section{Market Model Without Intermediate Consumption}\label{section1}
\subsection{Market Model and Mathematical Set Up}
We consider a financial market with one cash account and $d$ risky assets. The cash account is assumed to satisfy $S^0_{t}\equiv1, \forall t\in[0,T]$, which serves as the num\'{e}raire. Risky assets are modeled by a $d$-dimensional strictly positive process $(S_t)_{t\in[0,T]}=(S_t^1, \ldots, S_t^d)_{t\in[0,T]}$ on a given filtered probability space $(\Omega, \mathcal{F}, \mathbb{F}=(\mathcal{F}_{t})_{t\in [0,T]}, \mathbb{P})$, where the filtration $\mathbb{F}$ satisfies the usual conditions. The maturity time is given by $T$.

$(S_t)_{t\in[0,T]}$ may not be a semimartingale in general. To simplify our notation, we take $\mathcal{F}=\mathcal{F}_{T}$. Trading the risky assets incurs transaction costs. We define a nonnegative $(1+d)\times (1+d)$-matrix $\Lambda=(\lambda^{ij})_{0\leq i,j\leq d}$ with each $\lambda^{ij}\geq 0$ to model the proportional factor of costs that one has to pay if exchanging the $i$-th into the $j$-th asset. Clearly, $\lambda^{i0}=0$ for any $0\leq i\leq d$ as $S^0$ is the cash account. It is natural (see \cite{MAFI:MAFI180} and \cite{Campi06}) to impose that for $0\leq i,j,k\leq d$ the following holds
\begin{equation}
(1+\lambda^{ij})\leq (1+\lambda^{ik})(1+\lambda^{kj}).\nonumber
\end{equation}
The transaction cost coefficients $\Lambda$ may be constant or may depend on $t$ and $\omega$ in an adapted way. In this paper, we make the same assumption as in \cite{Campi06} that the bid-ask process
\begin{equation}
\pi_t^{ij}(\omega)\triangleq (1+\lambda_t^{ij}(\omega))\frac{S_t^{j}(\omega)}{S_t^{i}(\omega)},\ \ 0\leq i,j\leq d,\nonumber
\end{equation}
is \cadlag for all $0\leq i,j\leq d$. The $(1+d)\times (1+d)$ matrix $\Pi=(\pi^{i j})_{0\leq i,j\leq d}$ is called a bid-ask matrix.

We are working in the Kabanov's framework, which is centered on the idea of cone-valued processes. The solvency cone $\hat{K}$ is defined as a convex polyhedral cone in $\mathbb{R}^{1+d}$ spanned by the unit vectors $e^i$, $0\leq i\leq d$ and vectors $(1+\lambda^{ij})\frac{S^{j}}{S^i}e^i-e^j$, $0\leq i,j\leq d$. The convex cone $-\hat{K}$ should be interpreted as those portfolios available at price zero. A cone $\hat{K}$ is called \textit{proper} if $\hat{K}\cap(-\hat{K})=\{0\}$. In this paper, we shall assume that the cones $\hat{K}_t$ and $\hat{K}_{t-}$ are proper and contain $\mathbb{R}_+^{1+d}$ (efficient friction). In addition, we make the assumption that $\mathcal{F}_T=\mathcal{F}_{T-}$ and $\Pi_T=\Pi_{T-}$ a.s.. The cones $(\hat{K}_t)_{t\in[0,T]}$ induce a natural order among $\mathbb{R}^{1+d}$-valued random variables. In particular, for any stopping time $\tau$ and let $X$, $Y$ be two $\mathcal{F}_{\tau}$-measurable random variables. We denote $X\succeq_{\tau} Y$ if $X-Y\in\mathbb{L}^0(\hat{K}_{\tau}, \mathcal{F}_{\tau})$. Here we define $\mathbb{L}^0(\hat{K}_{\tau}, \mathcal{F}_{\tau})$, short as $\mathbb{L}^0(\hat{K}_{\tau})$, the cone of all $\hat{K}_{\tau}$-valued $\mathcal{F}_{\tau}$-measurable random variables. Following the previous notations, it is clear that $\mathbb{L}^0(\mathbb{R}_+^{1+d})\subset\mathbb{L}^0(\hat{K}_T)$.

Given a cone $\hat{K}$ in $\mathbb{R}^{1+d}$, its positive polar cone is defined by
\begin{equation}
\hat{K}^{\ast}\triangleq\{w\in\mathbb{R}^{1+d}:  \langle v,w\rangle \geq 0,\ \forall v\in \hat{K}\}.\nonumber
\end{equation}

\begin{definition}
An adapted $\mathbb{R}_{+}^{1+d}\setminus\{0\}$-valued, c\`{a}dl\`{a}g process $Z=(Z_{t}^0, Z_t^1,\ldots, Z_t^d)_{t\in[0,T]}$ with $Z_{0}^{0}=1$ is called a num\'{e}raire-based \textit{consistent price system} for the transaction costs $\Lambda$ if $Z^0$ is a martingale, $Z^i$ is a local martingale for $i=1,\ldots, d$ and $Z_{t}\in \hat{K}_{t}^{\ast}$ a.s. for every $t\in[0,T]$. Moreover, $Z$ will be called a num\'{e}raire-based \textit{strictly consistent price system} if for every $[0,T]\cup\{\infty\}$-valued stopping time $\tau$, $Z_{\tau}\in\text{int}(\hat{K}_{\tau}^{\ast})$ a.s. on $\{\tau<\infty\}$ and for every predictable $[0,T]\cup\{\infty\}$-valued stopping time $\sigma$, $Z_{\sigma-}\in\text{int}(\hat{K}_{\sigma-}^{\ast})$ a.s. on $\{\sigma<\infty\}$. The set of all num\'{e}raire-based consistent price systems (resp. strictly consistent price systems) will be denoted by $\mathcal{Z}$ (resp. $\mathcal{Z}^{s}$). For simplicity, we write (S)CPS to mean num\'{e}raire-based (strictly) consistent price system.
\end{definition}

In this paper, we will make the standing assumption:
\begin{assumption}\label{assumZ}
Existence of a SCPS: $\mathcal{Z}^{s}\neq \emptyset$.
\end{assumption}

\begin{remark}\label{rmkCPS}
Equivalently, each SCPS can be represented by a pair $(\mathbb{Q},\tilde{S})$ where $\mathbb{Q}$ is equivalent to $\mathbb{P}$ and $(\tilde{S}_t)_{t\in[0,T]}=(\tilde{S}_t^1, \ldots, \tilde{S}_t^d)_{t\in[0,T]}$ is a $d$-dimensional local martingale under $\mathbb{Q}$, see \cite{MAFI:MAFI180}, \cite{MAFI:MAFI004} and \cite{Kab}. $(\mathbb{Q},\tilde{S})$ is related to $Z\in\mathcal{Z}^s$ by setting $\frac{d\mathbb{Q}}{d\mathbb{P}}=Z^0_T$ and $\tilde{S}^i_t=Z_t^i/Z_t^0$ for $i=1,\ldots, d$ and $t\in[0,T]$.
\end{remark}

The following definitions are based on the above equivalent representation.
\begin{definition}
Denote $\mathcal{S}^s:=\{(Z^1/Z^0,\ldots, Z^d/Z^0): (Z^0,Z^1,\ldots, Z^d)\in\mathcal{Z}^s\}$. For each fixed $\tilde{S}\in\mathcal{S}^s$, we define
\begin{equation}
\mathcal{M}^s(\tilde{S})\triangleq \Big\{\mathbb{Q}: \frac{d\mathbb{Q}}{d\mathbb{P}}=Z^0_T\ \text{where}\ \Big(\frac{Z^1}{Z^0}, \ldots, \frac{Z^d}{Z^0}\Big)=\tilde{S},\ Z\in\mathcal{Z}^s\Big\}.\nonumber
\end{equation}
Similarly, for each fixed $\tilde{S}\in\mathcal{S}^s$, we define
\begin{equation}
\mathcal{Z}^s(\tilde{S})\triangleq\Big\{Z=(Z^0, Z^1,\ldots, Z^d):\ Z\in\mathcal{Z}^s\ \text{where}\ \Big(\frac{Z^1}{Z^0}, \ldots, \frac{Z^d}{Z^0}\Big)=\tilde{S}\Big\}.\nonumber
\end{equation}
Clearly, indexed by $\tilde{S}\in\mathcal{S}^s$, $\mathcal{Z}^s(\tilde{S})$ can be regarded as a partition of the set $\mathcal{Z}^s$ and $\mathcal{M}^s(\tilde{S})=\{\mathbb{Q}:\frac{d\mathbb{Q}}{d\mathbb{P}}=Z^0_T,\ \text{where}\ Z\in\mathcal{Z}^s(\tilde{S})\}$.  We also denote $\mathcal{M}^s\triangleq \bigcup_{\tilde{S}\in\mathcal{S}^s}\mathcal{M}^s(\tilde{S})$ and it follows that $\mathcal{M}^s=\{\mathbb{Q}: \frac{d\mathbb{Q}}{d\mathbb{P}}=Z_T^0,\ \text{where}\ Z\in\mathcal{Z}^s\}$.
\end{definition}

From now on, the market is enlarged by allowing trading $N$ European contingent claims at time $t=0$ with the final cash payoff $\mathcal{E}_T=(\mathcal{E}_T^i)_{1\leq i\leq N}$. We denote $q=(q^i)_{1\leq i\leq N}$ as static holdings in contingent claims $\mathcal{E}_T$. By allowing $q$ to take negative values, without loss of generality, we can assume that $\mathcal{E}_T^i\geq 0$ for $1\leq i\leq N$. Each $\mathcal{E}^i_T$ may be unbounded, however, it is assumed throughout the paper that $\sum_{i=1}^{N}\mathcal{E}_T^i$ is integrable uniformly with respect to all SCPS $\mathcal{Z}^s$ in the following sense:
\begin{assumption}\label{endowassum}
\begin{equation}\label{unic}
\lim_{m\rightarrow\infty}\sup_{Z\in\mathcal{Z}^s}\mathbb{E}\bigg[\Big\langle \Big(\Big(\sum_{i=1}^{N}\mathcal{E}_T^i\Big)\mathbf{1}_{\{\sum_{i=1}^N\mathcal{E}_T^i>m\}},\bar{\mathbf{0}}\Big), Z_T\Big\rangle\bigg] =0,
\end{equation}
where $\bar{\mathbf{0}}$ is the $d$-dimensional zero vector.
\end{assumption}
\begin{remark}
The condition $(\ref{unic})$ implies the finite super-hedging price of the random endowments under SCPS, i.e., $\sup_{Z\in\mathcal{Z}^s}\mathbb{E}[\langle (\sum_{i=1}^N\mathcal{E}_T^i, \bar{\mathbf{0}}), Z_T\rangle] =\sup_{\mathbb{Q}\in\mathcal{M}^s}\mathbb{E}^{\mathbb{Q}}[\sum_{i=1}^{N}\mathcal{E}_T^i]<\infty$. If we require that $\mathcal{E}_T\in\mathbb{L}^{\infty}$, Assumption $\ref{endowassum}$ holds trivially. In the present work, Assumption $\ref{endowassum}$ guarantees the super-hedging result holds for all $Z\in\mathcal{Z}^s$.
\end{remark}

To deal with unbounded random endowments, the set of \textit{admissible} portfolios with constant lower bounds is generally too small and is required for an extension to the set of \textit{acceptable} portfolios with some stochastic thresholds in the cash account. To fit into the framework with transaction costs, we need to modify the definition of acceptable portfolio (See \cite{DS97} and \cite{kram04} in the frictionless market) by taking into account of all SCPS. Since each $\tilde{S}\in\mathcal{S}^s$ is a $\mathbb{Q}$-local martingale, it follows that $\tilde{S}$ is a semimartingale under the physical probability measure $\mathbb{P}$. Given the initial wealth $a>0$, for each $\mathbb{P}$-semimartingale $\tilde{S}\in\mathcal{S}^s$, let $\mathcal{X}(\tilde{S}, a)$ be the set of nonnegative wealth processes in the $\tilde{S}$-market. That is,
\begin{equation}
\begin{split}
\mathcal{X}(\tilde{S}, a)=\{X\geq 0: &X_t=a+(H\cdot \tilde{S})_t,\ \text{where}\ H \text{ is predictable}\\
&\text{and $\tilde{S}$-integrable}, \ t\in[0,T]\}.\nonumber
\end{split}
\end{equation}
A wealth process in $\mathcal{X}(\tilde{S},a)$ is called \textit{maximal}, denoted by $X^{\max,\tilde{S}}$, if its terminal value $X_T^{\max,\tilde{S}}$ can not be dominated by that of any other processes in $\mathcal{X}(\tilde{S}, a)$.

\begin{lemma}\label{ran}
Under Assumption $\ref{endowassum}$, there exists a constant $a>0$ such that for each $\tilde{S}\in\mathcal{S}^s$, there exists a maximal element $\hat{X}^{\max, \tilde{S}}\in\mathcal{X}(\tilde{S},a)$ and $\sum_{i=1}^{N}\mathcal{E}_T^i\leq \hat{X}^{\max,\tilde{S}}_T$.
\end{lemma}

\begin{definition}\label{def111}
Given Assumption $\ref{assumZ}$, an $\mathbb{R}^{1+d}$-valued process $V=(V_{t}^{0},V_{t}^{1}, \ldots, V_t^d)_{t\in[0,T]}$ is called a \textbf{self-financing} portfolio process (see \cite{Campi06}) with transaction costs $\Lambda$ if it satisfies the following properties:
\begin{itemize}
\item[(1)] $V$ is predictable and $\mathbb{P}$-a.e. path of $V$ has finite variation
\item[(2)] For every pair of stopping times $0\leq \sigma\leq \tau\leq T$, we have
\begin{equation}\label{selfcond}
V_{\tau}-V_{\sigma}\in -\hat{\mathcal{K}}_{\sigma,\tau}\ \ \ \mathbb{P}-\text{a.s.}
\end{equation}
where
\begin{equation}
\hat{\mathcal{K}}_{\sigma, \tau}(\omega)\triangleq \overline{\text{conv}}\bigg(\bigcup_{\sigma(\omega)\leq t<\tau(\omega)}\hat{K}_{t} (\omega)\bigg)\nonumber
\end{equation}
with the bar closure taken in $\mathbb{R}^{1+d}$ for each $\omega\in\Omega$.
\end{itemize}
\end{definition}

\begin{remark}
As pointed out in R\'{a}sonyi's example in \cite{Ras}, since the bid-ask processes are assumed to be \cadlag, we have to take care of both left and right jumps of portfolio processes $(V_t)_{t\in[0,T]}$. Therefore, predictable processes are natural choices to model $(V_t)_{t\in[0,T]}$. If the stopping time $\tau$ is predictable, three values $V_{\tau-}$, $V_{\tau}$ and $V_{\tau+}$ can be different. $\triangle V_{\tau}=V_{\tau}-V_{\tau-}$ and $\triangle_+ V_{\tau}=V_{\tau+}- V_{\tau}$ can model a jump immediately before time $\tau$ and a jump at time $\tau$ respectively. The condition that $\mathbb{P}$-a.e. path of $V$ has finite variation is made for some mathematical convenience. But the financial intuition behind it is that a portfolio process having trajectories with infinite variation would lead to infinitely large transaction costs paid by the investor. Consequently, we only consider self-financing portfolios given in Definition $\ref{def111}$.
\end{remark}

\begin{definition}
A self-financing portfolio $V$ is called \textbf{admissible} (num\'{e}raire-based) if it additionally satisfies (see Definition $13$ of \cite{Pa})
\begin{itemize}
\item[(3)] There is a threshold $a>0$ such that $V_{T}+(a,\bar{\mathbf{0}})\in\mathbb{L}^0(\hat{K}_T)$ and $\langle V_{\tau}+(a, \bar{\mathbf{0}}), Z_{\tau}\rangle\geq 0$ $\mathbb{P}$-a.s. for all $[0,T]$-valued stopping time $\tau$ and for every SCPS $Z\in\mathcal{Z}^{s}$, where $\bar{\mathbf{0}}$ is the $d$-dimensional zero vector.
\end{itemize}
Instead, a self-financing portfolio $V$ is called \textbf{acceptable} (num\'{e}raire-based) if it satisfies the following property
\begin{itemize}
\item[(3')] There exists a constant $a>0$ such that for each $\tilde{S}\in\mathcal{S}^s$, there exists a maximal element $X^{\max, \tilde{S}}\in\mathcal{X}(\tilde{S}, a)$ with $V_T+(X^{\max,\tilde{S}}_T,\bar{\mathbf{0}})\in\mathbb{L}^0(\hat{K}_T)$ and $\langle V_{\tau}+(X^{\max,\tilde{S}}_{\tau},\bar{\mathbf{0}}), Z_{\tau}\rangle\geq 0$ $\mathbb{P}$-a.s. for all $[0,T]$-valued stopping time $\tau$ and for every $Z\in\mathcal{Z}^s(\tilde{S})$.
\end{itemize}
Denote the set of all acceptable dynamic portfolio processes by $\mathcal{V}^{\text{acpt}}$ and set
\begin{equation}
\mathcal{V}_x^{\text{acpt}}\triangleq \{V\in\mathcal{V}^{\text{acpt}}: V_0=x=(x^0,x^1,\ldots, x^d)\},\nonumber
\end{equation}
for some initial position $x\in\mathbb{R}^{1+d}$.
\end{definition}

\begin{remark}
It is clear that every admissible portfolio process is acceptable since each constant $a>0$ is a maximal element in $\mathcal{X}(\tilde{S},a)$. To see this, we notice that for each $\tilde{S}\in\mathcal{S}^s$, there exists $\mathbb{Q}\sim\mathbb{P}$ such that $\tilde{S}$ is a $\mathbb{Q}$-local martingale. Therefore we can conclude that each $\tilde{S}$ is a semimartingale, see Theorem $7.2$ of \cite{Sch94} for locally bounded $\tilde{S}$ and Theorem $1.3$ of \cite{Kardaras20112678} for nonnegative $\tilde{S}$. In addition, the semimartingale $\tilde{S}$ satisfies No Free Lunch with Vanishing Risk condition, see Theorem $1.1$ of \cite{MR1671792}. Hence, there is no maximal element in $\mathcal{X}(\tilde{S},a)$ which dominates the constant $a$.
\end{remark}

However, the definition above seems abstract in general since it involves all SCPS. In the next examples, we will provide some conditions easy to check such that the self-financing portfolio $V$ is indeed acceptable:
\begin{example}
If there exists a nonnegative $\mathbb{P}$-martingale $W$ with $a\triangleq \sup_{\mathbb{Q}\in\mathcal{M}^s}\mathbb{E}^{\mathbb{Q}}[W_T]<\infty$ such that $V_T+(W_T, \bar{\mathbf{0}})\in\mathbb{L}^0(\hat{K}_T)$ and $\langle V_{\tau}+(W_{\tau},\bar{\mathbf{0}}), Z_{\tau}\rangle\geq 0$ $\mathbb{P}$-a.s. for all $[0,T]$-valued stopping times $\tau$ and for all SCPS $Z\in\mathcal{Z}^{s}$, we can conclude that the self-financing portfolio $V$ is acceptable. To see this, by following the same argument of Lemma $\ref{ran}$, we get that for each $\tilde{S}\in\mathcal{S}^s$, there exists a maximal element $X^{\max,\tilde{S}}\in\mathcal{X}(\tilde{S},a)$ such that $W_T\leq X_T^{\max,\tilde{S}}$. Moreover, for each $\tilde{S}\in\mathcal{S}^s$, we also know that $(X_t^{\max,\tilde{S}}-W_t)_{t\in[0,T]}$ is a supermartingale under $\mathbb{P}$ since each $X^{\max,\tilde{S}}$ is a $\mathbb{P}$-supermartingale. It follows that $X_{\tau}^{\max,\tilde{S}}-W_{\tau}\geq 0$ $\mathbb{P}$-a.s. for all $[0,T]$-valued stopping times $\tau$. By passing the inequality to $V$ and the fact that $\mathbb{L}^0(\mathbb{R}_+^{1+d})\subset\mathbb{L}^0(\hat{K}_T)$, we can verify that all conditions of the acceptable portfolio are satisfied.
\end{example}

\begin{example}
Let $S$ be a strictly positive semimartingale and the transaction cost $\lambda^{ij}=\lambda\in(0,1)$ be a fixed constant for all $0\leq i\leq d$ and $1\leq j\leq d$. It is well known that for each $\tilde{S}\in\mathcal{S}^s$, we have
\begin{equation}
(1-\lambda)S_t^{i}<\tilde{S}_t^i< (1+\lambda)S_t^i,\ \ \mathbb{P}-\text{a.s.}\ \ 1\leq i\leq d,\ \ t\in[0,T].\nonumber
\end{equation}
Therefore, for any fixed constant $k>0$, we obtain that for any constant $a> (1-\lambda)k\sum_{i=1}^{d}S_0^i$ and all $[0,T]$-valued stopping time $\tau$,
\begin{equation}
0<a+\int_0^{\tau}(1-\lambda)k\bar{\mathbf{1}}dS_u <a+\int_0^{\tau}k\bar{\mathbf{1}}d\tilde{S}_u,\nonumber
\end{equation}
where $\bar{\mathbf{1}}$ is the d-dimensional identity vector, which leads to
\begin{equation}
0<a+(1-\lambda)k\sum_{i=1}^{d}(S_{\tau}^i-S_0^i)<a+\int_0^{\tau}k\bar{\mathbf{1}}d\tilde{S}_u.\nonumber
\end{equation}
For each fixed $\tilde{S}\in\mathcal{S}^s$, we can choose $X^{\max, \tilde{S}}\in\mathcal{X}(\tilde{S},a)$ such that $X_T^{\max,\tilde{S}}\geq a+\int_0^{T}k\bar{\mathbf{1}}d\tilde{S}_u$. It is easy to see that $(a+\int_0^{t}k\bar{\mathbf{1}}d\tilde{S}_u)_{t\in[0,T]}$ is a local martingale, hence supermartingale, under all $\mathbb{Q}\in\mathcal{M}^s(\tilde{S})$. Meanwhile, by Theorem $5.7$ of \cite{MR1671792}, for any $X^{\max, \tilde{S}}\in\mathcal{X}(\tilde{S},a)$ there exists some $\mathbb{Q}^{\ast}\in\mathcal{M}^s(\tilde{S})$ such that $X^{\max, \tilde{S}}$ is a UI martingale under $\mathbb{Q}^{\ast}$. Thus, it follows that $X^{\max, \tilde{S}}_{\tau}\geq a+\int_0^{\tau}k\bar{\mathbf{1}}d\tilde{S}_u$ $\mathbb{Q}^{\ast}$-a.s. for all $[0,T]$-valued stopping times. Since $\mathbb{Q}^{\ast}$ is equivalent to $\mathbb{P}$, we deduce that $X^{\max, \tilde{S}}_{\tau}\geq a+\int_0^{\tau}k\bar{\mathbf{1}}d\tilde{S}_u$ holds $\mathbb{P}$-a.s.

Let us denote $B_t\triangleq a+(1-\lambda)k\sum_{i=1}^{d}(S_{t}^i-S_0^i)$ for $t\in[0,T]$. It follows that for each $\tilde{S}\in\mathcal{S}^s$, there exists an $X^{\max, \tilde{S}}\in\mathcal{X}(\tilde{S},a)$ such that
\begin{equation}
B_{\tau}<X_{\tau}^{\max,\tilde{S}},\ \ \mathbb{P}-\text{a.s.}.\nonumber
\end{equation}
Therefore, in this market, if the self-financing portfolio $V$ satisfies the conditions:  $V_T+(B_T, \bar{\mathbf{0}})\in\mathbb{L}^0(\hat{K}_T)$ and $\langle V_{\tau}+(B_{\tau},\bar{\mathbf{0}}), Z_{\tau}\rangle\geq 0$ $\mathbb{P}$-a.s. for all $[0,T]$-valued stopping time $\tau$ and for all SCPS $Z\in\mathcal{Z}^{s}$, the above argument implies that $V$ is an acceptable portfolio process.
\end{example}

We now proceed to show a type of super-hedging theorem for some workable contingent claims using acceptable portfolios. The next assumption is required to exclude some trivial cases.

\begin{assumption}\label{notrep}
For any non-zero vector $q\in\mathbb{R}^N$, the random variable $q\cdot\mathcal{E}_T$ is not replicable in the market under SCPS.
\end{assumption}

Denote $\mathcal{H}(x,q)$ the set of acceptable portfolio processes with initial position $x\in\mathbb{R}^{1+d}$ whose terminal value dominate the payoff $(-q\cdot\mathcal{E}_T,\bar{\mathbf{0}})$, i.e.,
\begin{equation}\label{newacp}
\mathcal{H}(x,q)\triangleq \{V: V_T+(q\cdot\mathcal{E}_T,\bar{\mathbf{0}})\in\mathbb{L}^0(\hat{K}_T),\ V\in\mathcal{V}_x^{\text{acpt}}\},\ \ (x,q)\in\mathcal{K}
\end{equation}
where the effective domain $\mathcal{K}$ is defined by
\begin{equation}\label{workK}
\mathcal{K}\triangleq \text{int}\{(x,q)\in\mathbb{R}^{1+d+N}: \mathcal{H}(x,q)\neq \emptyset\}.
\end{equation}

\ \\
\subsection{Superhedging Result for Some Workable Contingent Claims}
Let us consider the abstract set
\begin{equation}
\mathcal{C}(x,q)\triangleq \{g\in\mathbb{L}_+^{0}(\mathbb{R}^{1+d}): V_T+(q\cdot\mathcal{E}_T,\bar{\mathbf{0}})-g\in\mathbb{L}^0(\hat{K}_T),\  V\in\mathcal{H}(x,q)\},\ \ (x,q)\in\mathcal{K}.\nonumber
\end{equation}

The following result first gives a characterization of the elements in $\mathcal{C}(x,q)$ using all SCPS.
\begin{lemma}\label{budlem}
Given Assumption $\ref{assumZ}$, $\ref{endowassum}$, if $(x,q)\in\mathcal{K}$, for any $g\in\mathcal{C}(x,q)$, we have
\begin{equation}\label{leminq}
\mathbb{E}[\langle g, Z_T\rangle]\leq \langle x, Z_0\rangle+\mathbb{E}[\langle (q\cdot\mathcal{E}_T,\bar{\mathbf{0}}), Z_T\rangle],\ \ \forall Z\in\mathcal{Z}^s.
\end{equation}
\end{lemma}

\begin{remark}
In \cite{kram04} without transaction costs, the authors only require the super-hedging condition of the random endowments, i.e., $\sup_{\mathbb{Q}\in\mathcal{M}}\mathbb{E}[|\mathcal{E}_T^i|]<\infty$ for all $1\leq i\leq N$. Due to the special property of stochastic integrals, they can prove the inequality similar to $(\ref{budlem})$ with the subset $\mathcal{M}'$ of all equivalent local martingale measures $\mathcal{M}$, which depends on the random endowments $\mathcal{E}_T$, see Lemma $4$ and Lemma $5$ in \cite{kram04}. Moreover, the subset $\mathcal{M}'$ is enough for them to build the bipolar results for conjugate duality.

In models with transaction costs, the definition of acceptable portfolios is more delicate. We can similarly define a subset $\bar{\mathcal{M}}^s(\tilde{S})$ using the given random endowment such that each $\hat{X}^{\max,\tilde{S}}$ found in Lemma $\ref{ran}$ is a true martingale under $\mathbb{Q}\in\bar{\mathcal{M}}^s(\tilde{S})$. However, this set is no longer appropriate for our model since each acceptable portfolio needs some different lower bound $X^{\max,\tilde{S}}$ for all $[0,T]$-stopping time $\tau$, and $X^{\max,\tilde{S}}$ is not necessarily a martingale under $\bar{\mathcal{M}}^s(\tilde{S})$.

To deal with this issue, we need to avoid the subset trick in \cite{kram04} and prove the super-hedging result for the whole set $\mathcal{Z}^s$. To this end, we have to make the Assumption $\ref{endowassum}$ which is stronger than the conventional super-hedging requirement that $\sup_{Z\in\mathcal{Z}^s}\mathbb{E}\Big[\Big\langle\Big(\sum_{i=1}^N|\mathcal{E}_T^i|, \bar{\mathbf{0}}\Big), Z_T\Big\rangle\Big] <\infty$.
\end{remark}

The next result, on the other hand, gives a criteria to check if the given process is in the set $\mathcal{C}(x,q)$ or not.
\begin{lemma}\label{lemsp}
Let Assumption $\ref{assumZ}$ hold. Let $g$ be an $\mathbb{R}^{1+d}$-valued, $\mathcal{F}_T$-measurable random vector such that there exists a constant $a>0$ and for each $\tilde{S}\in\mathcal{S}^s$, there exists a maximal element $X^{\max, \tilde{S}}\in\mathcal{X}(\tilde{S},a)$ such that $g+(X_T^{\max,\tilde{S}},\bar{\mathbf{0}})\in \mathbb{L}^0(\hat{K}_T)$. If
\begin{equation}\label{superhedging_ineq}
\mathbb{E}[\ab{g, Z_T}]\leq \langle x, Z_0\rangle,\ \ \forall \, Z\in\mathcal{Z}^s,
\end{equation}
where $x\in\mathbb{R}^{1+d}$, there exists $V\in\mathcal{V}_x^{\text{acpt}}$ such that $V_T-g\in\mathbb{L}^0(\hat{K}_T)$.
\end{lemma}

\begin{remark}
Lemma $\ref{budlem}$ and Lemma $\ref{lemsp}$ together provide a super-hedging theorem for some workable contingent claims using acceptable portfolios. Without the intermediate consumption, we can also study the multi-variate utility maximization problem on the terminal wealth defined on each asset with unbounded random endowments similar to \cite{campi1}. Moreover, using the acceptable portfolios, we can possibly perform the sensitivity analysis of marginal utility-based prices with respect to a small number of random endowments, similar to \cite{Sirbu2006}, however under proportional transaction costs. Some potential extensions in these directions are scheduled as future research projects.
\end{remark}
\ \\

\section{Market Model With Consumption and Habit Formation}\label{section2}

\subsection{Set Up}
In this section, we adopt the financial market model with proportional transaction costs and unbounded random endowment as in Section $\ref{section1}$. In addition, we start to assume that the agent can also choose an intermediate consumption from the cash account $S^0$ during the investment horizon. The consumption rate process is denoted by $(c_t)_{t\in[0,T]}$. To simplify the notation, starting from $t=0$, we assume that the investor holds initial wealth $V_0=(x, \bar{\mathbf{0}})$ with $x\in\mathbb{R}$, i.e., the initial position in the cash account is $x\in\mathbb{R}$ and the initial position in each risky asset account is $0$.

\begin{definition}
Given Assumption $\ref{assumZ}$, $\ref{endowassum}$ and $\ref{notrep}$ and let $(x,\bar{\mathbf{0}},q)\in\mathcal{K}$, the consumption process $(c_t)_{t\in[0,T]}$ is called $(x,q\cdot\mathcal{E})$-\textbf{financeable} if there exists a self-financing and acceptable portfolio $V\in \mathcal{H}(x,q)$, defined in $(\ref{newacp})$, such that $V_0=(x,\bar{\mathbf{0}})$ and $V_T+(-\int_0^Tc_tdt+q\cdot\mathcal{E}_T,\bar{\mathbf{0}})\in\mathbb{L}^0(\hat{K}_T)$. Denote $\mathcal{C}_{x,q\cdot\mathcal{E}_T}$ the set of all $(x,q\cdot\mathcal{E}_T)$-financeable consumption processes.
\end{definition}

\begin{proposition}[Consumption Budget Constraint]\label{bcccc}
Let Assumption $\ref{assumZ}$, $\ref{endowassum}$ and $\ref{notrep}$ hold. Let $(x,\bar{\mathbf{0}},q)\in\mathcal{K}$ which is defined in $(\ref{workK})$. We have that the process $c\in\mathcal{C}_{x,q\cdot\mathcal{E}_T}$ if and only if
\begin{equation}\label{budccc}
\mathbb{E}\Big[\int_0^Tc_tZ^0_tdt\Big]\leq x+\mathbb{E}[\langle (q\cdot\mathcal{E}_T, \bar{\mathbf{0}}), Z_T\rangle],\ \ \ \forall Z\in\mathcal{Z}^s.
\end{equation}
\end{proposition}

In the present paper, we are interested in the time non-separable preference on consumption which takes into account the path-dependence feature of the consumption behavior. In particular, we introduce the \textit{consumption habit formation} process $F(c)_t$ given by the exponentially weighted average of agent's past consumption integral and the initial habit
\begin{equation}
F(c)_t=ze^{-\int_0^t\alpha_vdv}+\int_0^t\delta_se^{-\int_s^t\alpha_vdv}c_sds,\ \ t\in[0,T],\nonumber
\end{equation}
where the constant $z\geq 0$ is called the \textit{initial habit}. In general, the discounting factors $\alpha$ and $\delta$ are assumed to be nonnegative optional processes which are allowed to be unbounded. However, for the concern of integrability, it is assumed that $\int_0^t(\delta_u-\alpha_u)du<\infty$ a.s. for each $t\in[0,T]$.

This paper is interested in the conventional scenario that the consumption habit is addictive in the sense that $c_t\geq F(c)_t$, $\forall t\in[0,T]$, i.e., the investor's current consumption rate shall never fall below the standard of living process.

The investor's preference is represented by a utility function $\mathit{U}:[0,T]\times (0, \infty)\rightarrow\mathbb{R}$, such that, for every $x>0$, $\mathit{U}(\cdot,x)$ is continuous on $[0,T]$, and for every $t\in[0,T]$, the function $\mathit{U}(t,\cdot)$ is strictly concave, strictly increasing, continuously differentiable and satisfies the Inada conditions:
\begin{equation}\label{Inada}
\mathit{U}'(t,0)\triangleq\underset{x\rightarrow 0}{\lim}\mathit{U}'(t,x)=\infty,\ \ \ \ \ \  \mathit{U}'(t,\infty)\triangleq\underset{x\rightarrow \infty}{\lim}\mathit{U}'(t,x)=0,
\end{equation}
where $\mathit{U}'(t,x)\triangleq\frac{\partial}{\partial x}\mathit{U}(t,x)$. For each $t\in[0,T]$, we extend the definition of the utility function by $\mathit{U}(t,x)=-\infty$ for all $x<0$, which is equivalent to the addictive habit formation constraint $c_t\geq F(c)_t$. The convex conjugate of the utility function, is defined by
\begin{equation}
\mathit{V}(t,y)\triangleq\underset{x>0}{\sup}\{\mathit{U}(t,x)-xy\},\ \ \ \ y>0.\nonumber
\end{equation}

Following \cite{Yu1} , we make assumptions on the asymptotic behavior of $\mathit{U}$ at both $x=0$ and $x=\infty$.
\noindent
\begin{assumption}\label{ASSUV}\ \\
The utility function $\mathit{U}$ satisfies the Reasonable Asymptotic Elasticity (RAE) condition both at $x=\infty$ and $x=0$, i.e.,
\begin{equation}\label{ass:AEU}
AE_{\infty}[\mathit{U}] =\underset{x\rightarrow\infty}{\lim\sup}\Big(\sup_{t\in[0,T]}\frac{x\mathit{U}'(t,x)}{\mathit{U}(t,x)}\Big)<1,
\end{equation}
and
\begin{equation}\label{ass:AEV}
AE_{0}[\mathit{U}] =\underset{x\rightarrow 0}{\lim\sup}\Big(\sup_{t\in[0,T]}\frac{x\mathit{U}'(t,x)}{|\mathit{U}(t,x)|}\Big)<\infty.
\end{equation}
Moreover, in order to get some inequalities uniformly in time $t$, we shall assume
\begin{equation}\label{up}
\lim_{x\rightarrow\infty}\Big(\inf_{t\in[0,T]}\mathit{U}(t,x)\Big)>0,
\end{equation}
and
\begin{equation}\label{down}
\lim_{x\rightarrow 0}\Big(\sup_{t\in[0,T]}\mathit{U}(t,x)\Big)<0.
\end{equation}
\end{assumption}

The RAE conditions $(\ref{ass:AEU})$ and $(\ref{ass:AEV})$ are not restrictive. For instance, the well known discounted log utility function $U(t,x)=e^{-\beta t}\log x$ and the discounted power utility function $U(t,x)=e^{-\beta t}\frac{x^p}{p}$ $(p<1\ \text{and}\ p\neq 0)$ satisfy the conditions $(\ref{ass:AEU})$ and $(\ref{ass:AEV})$. Actually, if the utility function has the finite lower bound condition $\inf_{t\in[0,T]}U(t,0)>-\infty$, the condition $(\ref{ass:AEV})$ is verified. On the other hand, it is also easy to check that the utility function $U(t,x)=-e^{\frac{1}{x}}$ does not satisfy the condition $(\ref{ass:AEV})$ and the utility function $U(t,x)=\frac{x}{\log x}$ does not satisfy the condition $(\ref{ass:AEU})$. Moreover, the extra conditions $(\ref{up})$ and $(\ref{down})$ are also not restrictive. Indeed, the utility function satisfies RAE conditions $(\ref{ass:AEU})$ and $(\ref{ass:AEV})$ if and only if its affine transform $a+bU(t,x)$ satisfies RAE conditions  $(\ref{ass:AEU})$ and $(\ref{ass:AEV})$ for arbitrary constants $a,b>0$.

As in \cite{Yu1}, we denote $\mathcal{O}$ as $\sigma$-algebra of optional sets relative to the filtration $(\mathcal{F}_{t})_{t\in[0,T]}$, and let $d\bar{\mathbb{P}}=dt\times d\mathbb{P}$ be the measure on the product space $(\Omega\times[0,T],\mathcal{O})$ defined as
\begin{equation}
\bar{\mathbb{P}}[A]=\mathbb{E}^{\mathbb{P}}\Big[\int_{0}^{T}\mathbf{1}_{A}(t,\omega)dt\Big],\ \ \ \textrm{for}\ A\in\mathcal{O}.\nonumber
\end{equation}
We denote by $\mathbb{L}^{0}(\Omega\times[0,T],\mathcal{O},\bar{\mathbb{P}})$ ($\mathbb{L}^{0}(\Omega\times[0,T])$ for short) the set of all random variables on the product space with respect to the optional $\sigma$-algebra $\mathcal{O}$ endowed with the topology of convergence in measure $\bar{\mathbb{P}}$. And from now on, we shall identify the optional stochastic process $(Y_t)_{t\in[0,T]}$ with the random variable $Y\in \mathbb{L}^{0}(\Omega\times[0,T])$. We also define the positive orthant $\mathbb{L}_{+}^{0}(\Omega\times[0,T],\mathcal{O},\bar{\mathbb{P}})$ ($\mathbb{L}_{+}^{0}(\Omega\times[0,T])$ for short) as the set of $Y=Y(t,\omega)\in\mathbb{L}^{0}$ such that
\begin{equation}
Y\geq 0,\ \ \ \ \bar{\mathbb{P}}-a.s. .\nonumber
\end{equation}

At this point, for any $(x,\bar{\mathbf{0}}, q)\in\mathcal{K}$ and any $z\geq 0$, we can define the set of all $(x, q\cdot\mathcal{E}_T)$-financeable consumption processes with habit formation constraint as a set of random variables on the product space by
\begin{equation}
\begin{split}
\mathcal{A}(x,q,z)\triangleq \Big\{&c\in\mathbb{L}_+^0(\Omega\times[0,T]): c_t\geq F(c)_t,\ \forall t\in[0,T]\ \text{and}\ \ c\in\mathcal{C}_{x,q\cdot\mathcal{E}_T}\Big\}\\
=\Big\{&c\in\mathbb{L}_{+}^{0}(\Omega\times[0,T]): c_t\geq F(c)_t,\ \forall t\in[0,T]\\
& \text{and}\ \mathbb{E}\Big[\int_0^Tc_tZ_t^0dt\Big]\leq x+\mathbb{E}[q\cdot\mathcal{E}_TZ_T^0],\ \ \forall Z\in \mathcal{Z}^s\Big\}.\nonumber
\end{split}
\end{equation}

However, the set $\mathcal{A}(x,q,z)$ may be empty for some values $(x,\bar{\mathbf{0}}, q)\in\mathcal{K}$ and $z\geq 0$ in virtue of the constraint. We shall restrict ourselves to the \textit{effective domain} $\bar{\mathcal{L}}$ which is defined as the union of the \textit{interior} of the set such that $\mathcal{A}(x,q,z)$ is not empty and the boundary $\{(x,q,z)\in\mathbb{R}^{N+2}:(x, \bar{\mathbf{0}}, q)\in\mathcal{K}\ \text{and}\ z=0\}$:
\begin{equation}
\begin{split}
\bar{\mathcal{L}}\triangleq \text{int}&\Big\{(x,q,z)\in\mathbb{R}^{N+2}: (x,\bar{\mathbf{0}},q)\in\mathcal{K},\ z>0\ \text{such that}\ \mathcal{A}(x,q,z)\neq \emptyset\Big\}\cup\\
& \{(x,q,z)\in\mathbb{R}^{N+2}:(x, \bar{\mathbf{0}}, q)\in\mathcal{K}\ \text{and}\ z=0\}.\nonumber
\end{split}
\end{equation}
From the definition, $\bar{\mathcal{L}}$ includes the special case of zero initial habit, i.e., $z=0$.

By choosing $(x,q,z)\in\bar{\mathcal{L}}$, we can now define the preliminary version of our \textbf{Primal Utility Maximization Problem} by
\begin{equation}\label{primalp}
u(x,q,z)\triangleq\underset{c\in\mathcal{A}(x,q,z)}{\sup}\mathbb{E}\Big[\int_{0}^{T}\mathit{U}(t,c_{t}-F(c)_t)dt\Big],\ \ \ (x,q,z)\in\bar{\mathcal{L}}.
\end{equation}

It is important to impose the following additional conditions on the discounting factors $\alpha_{t}$ and $\delta_{t}$, which are essential for the well-posedness of the primal utility optimization problem:
\begin{assumption}\label{asssss} \ \\
The nonnegative optional processes $(\alpha_t)_{t\in[0,T]}$ and $(\delta_t)_{t\in[0,T]}$ are assumed to satisfy
\begin{itemize}
\item[(i)] For any non-zero vector $(q,z)\in\mathbb{R}^{N+1}$, the random variable $-z\int_0^Te^{\int_0^t(\delta_v-\alpha_v)dv}dt+q\cdot\mathcal{E}_T$ is not replicable under SCPS.
\item[(ii)] We have
\begin{equation}\label{ass3.1}
\underset{Z\in\mathcal{Z}^s}{\sup}\mathbb{E}\Big[\int_{0}^{T}e^{\int_{0}^{t}(\delta_{v}-\alpha_{v})dv}Z^0_{t}dt\Big]<\infty.
\end{equation}
\item[(iii)] There exists a constant $\bar{x}>0$ such that
\begin{equation}\label{ass4.1}
\mathbb{E}\Big[\int_{0}^{T}\mathit{U}^{-}(t,\bar{x}e^{-\int_{0}^{t}\alpha_{v}dv})dt\Big] <\infty.
\end{equation}
\end{itemize}
\end{assumption}

\begin{remark}
If stochastic discounting processes $(\alpha_{t})_{t\in[0,T]}$ and $(\delta_{t})_{t\in[0,T]}$ are assumed to be bounded, conditions $(\ref{ass3.1})$ and $(\ref{ass4.1})$ will be satisfied. Condition $(\ref{ass3.1})$ is the well-known super-hedging property of the random variable $\int_0^Te^{\int_0^{t}(\delta_v-\alpha_v)dv}dt$ in the original market.
\end{remark}

The following Lemma gives an explicit characterization of the domain $\bar{\mathcal{L}}$.
\begin{lemma}\label{equ1}
Under Assumption $\ref{assumZ}$, $\ref{endowassum}$ and the condition $(\ref{ass3.1})$, the effective domain $\bar{\mathcal{L}}$ can be rewritten as
\begin{equation}
\bar{\mathcal{L}}=\Big\{ (x,q,z)\in\mathbb{R}^{N+2}: z\geq 0\ \text{and}\  x+\mathbb{E}[q\cdot\mathcal{E}_TZ_T^0]>z\mathbb{E}\Big[\int_0^Te^{\int_0^t(\delta_v-\alpha_v)dv}Z_t^0dt\Big],\ \forall Z\in\mathcal{Z}^s\Big\}.\nonumber
\end{equation}
\end{lemma}

\subsection{Path Dependence Reduction}

To deal with the path-dependence structure in the optimization problem, we will follow the trick in \cite{Yu1} and define the auxiliary process $\tilde{c}_t=c_t-F(c)_t$, $t\in[0,T]$. Denote the set of all auxiliary processes  by
\begin{equation}\label{comebackA}
\bar{\mathcal{A}}(x,q,z)\triangleq \{\tilde{c}\in\mathbb{L}_+^0(\Omega\times[0,T]): \tilde{c}_t=c_t-F(c)_t,\ \ \forall t\in[0,T],\ c\in\mathcal{A}(x,q,z)\}.
\end{equation}
The following lemma is a consequence of its definition.

\begin{lemma}
For each fixed $(x,q,z)\in\bar{\mathcal{L}}$, there is a one to one correspondence between sets $\mathcal{A}(x,q,z)$ and $\bar{\mathcal{A}}(x,q,z)$, and for all $(x,q,z)\in\bar{\mathcal{L}}$, we have $\bar{\mathcal{A}}(x,q,z)\neq \emptyset$.
\end{lemma}

For each SCPS $Z\in\mathcal{Z}^s$, we introduce the following important auxiliary optional process
\begin{equation}
\Gamma_t\triangleq Z^0_t+\delta_t\mathbb{E}\Big[\int_t^T e^{\int_t^s(\delta_v-\alpha_v)dv} Z_s^0ds\Big|\mathcal{F}_t\Big],\ \ \forall t\in[0,T],\nonumber
\end{equation}
and define the set of all these auxiliary processes by
\begin{equation}\label{initialdualM}
\widetilde{\mathcal{M}}\triangleq \Big\{\Gamma\in\mathbb{L}_+^0(\Omega\times[0,T]): \Gamma_t\triangleq Z^0_t+\delta_t\mathbb{E}\Big[\int_t^T e^{\int_t^s(\delta_v-\alpha_v)dv} Z_s^0ds\Big|\mathcal{F}_t\Big],\ \forall t\in[0,T],\ Z\in\mathcal{Z}^s \Big\}.
\end{equation}
Since stochastic discounting processes $\delta$ and $\alpha$ are assumed to be unbounded in general, under condition $(\ref{ass3.1})$, the auxiliary dual process $\Gamma$ is well defined, however, it is not necessarily integrable.

The next equivalent characterization of set $\bar{\mathcal{A}}(x,q,z)$ is crucial to reduce the path dependence feature and embed our problem into an auxiliary abstract optimization problem on the product space.

\begin{lemma}\label{equivA}
For $(x,q,z)\in\bar{\mathcal{L}}$, we can rewrite $\bar{\mathcal{A}}(x,q,z)$ as
\begin{equation}
\bar{\mathcal{A}}(x,q,z)=\Big\{\tilde{c}\in\mathbb{L}_+^0(\Omega\times[0,T]):\mathbb{E}\Big[\int_0^T\tilde{c}_t\Gamma_tdt\Big]\leq x-z\mathbb{E}\Big[\int_0^T\tilde{w}_t\Gamma_tdt\Big]+\mathbb{E}[q\cdot\mathcal{E}_T\Gamma_T],\ \forall \Gamma\in\widetilde{\mathcal{M}}\Big\},\nonumber
\end{equation}
where $\tilde{w}_t\triangleq e^{\int_0^t(-\alpha_v)dv}$ for $t\in[0,T]$.
\end{lemma}

In order to build the conjugate duality, we need to enlarge the effective domain $\bar{\mathcal{L}}$ to a natural domain, otherwise the constraint on the domain will affect the definition of the correct dual problem. First, the primal set $\bar{\mathcal{A}}(x,q,z)$ needs to be enlarged to the following abstract version
\begin{equation}
\widetilde{\mathcal{A}}(x,q,z)\triangleq \Big\{\tilde{c}\in\mathbb{L}_+^0(\Omega\times[0,T]): \mathbb{E}\Big[\int_0^T\tilde{c}_t\Gamma_tdt\Big]\leq x-z\mathbb{E}\Big[\int_0^T\tilde{w}_t\Gamma_tdt\Big]+\mathbb{E}[q\cdot\mathcal{E}_T\Gamma_T],\ \forall \Gamma\in\widetilde{\mathcal{M}}\Big\},\nonumber
\end{equation}
where now $(x,q,z)\in\mathbb{R}^{N+2}$. Second, we need to consider the enlarged domain $\mathcal{L}$
\begin{equation}
\mathcal{L}\triangleq \text{int}\{(x,q,z)\in\mathbb{R}^{N+2}:\ \widetilde{\mathcal{A}}(x,q,z)\neq\emptyset\}.\nonumber
\end{equation}

The next result shows that the set $\mathcal{L}$ is indeed the enlargement of the effective domain $\bar{\mathcal{L}}$.

\begin{lemma}\label{equ2}
We can characterize the set $\mathcal{L}$ equivalently by
\begin{equation}
\mathcal{L}=\text{int}\Big\{(x,q,z)\in\mathbb{R}^{N+2}: x+\mathbb{E}[q\cdot\mathcal{E}_TZ_T^0]-z\mathbb{E}\Big[\int_0^Te^{\int_0^t(\delta_v-\alpha_v)dv}Z_t^0dt\Big]\geq 0,\ \forall Z\in\mathcal{Z}^s \Big\}.\nonumber
\end{equation}
\end{lemma}

Based on the abstract primal set $\widetilde{\mathcal{A}}(x,q,z)$, we define the \textbf{Auxiliary Primal Utility Maximization Problem} by
\begin{equation}\label{newpro}
\tilde{u}(x,q,z)\triangleq \sup_{\tilde{c}\in\widetilde{\mathcal{A}}(x,q,z)}\mathbb{E}\Big[\int_0^T U(t,\tilde{c}_t)dt\Big],\ \ (x,q,z)\in\mathcal{L}.
\end{equation}

From the definition of $\bar{\mathcal{A}}(x,q,z)$ for $(x,q,z)\in\bar{\mathcal{L}}$ and $\widetilde{\mathcal{A}}(x,q,z)$ for $(x,q,z)\in\mathcal{L}$, Lemma $\ref{equ1}$ and Lemma $\ref{equ2}$ imply that $\bar{\mathcal{L}}\subset\mathcal{L}$. If we restrict $(x,q,z)\in\bar{\mathcal{L}}\subset\mathcal{L}$, the following equivalence holds
\begin{equation}
\bar{\mathcal{A}}(x,q,z)=\widetilde{\mathcal{A}}(x,q,z).\nonumber
\end{equation}
The equivalence between value functions follows, i.e.,
\begin{equation}
u(x,q,z)=\tilde{u}(x,q,z).\nonumber
\end{equation}
Moreover, $(c_t^{\ast})_{t\in[0,T]}$ is the optimal solution for $u(x,q,z)$ if and only if $(\tilde{c}_t^{\ast})_{t\in[0,T]}=(c_t^{\ast}-F(c^{\ast})_t)_{t\in[0,T]}$ is the optimal solution for $\tilde{u}(x,q,z)$. Therefore, we embedded our path-dependent utility maximization problem $(\ref{primalp})$ into the auxiliary abstract utility maximization problem $(\ref{newpro})$ without habit formation, however with the additional shadow random endowment $\tilde{w}$.\\
\ \\

\section{The Dual Problem and Main Results}\label{section3}
Similar to \cite{kram04} and \cite{Yu1}, we first introduce the set $\mathcal{R}$
\begin{equation}
\mathcal{R}\triangleq \text{ri}\Big\{(y,r)\in\mathbb{R}^{N+2}: xy+(-z,q)\cdot r\geq 0\ \text{for all}\ (x,q,z)\in\mathcal{L}\Big\},\nonumber
\end{equation}
where $(-z,q)\cdot r\triangleq -zr^0+\sum_{i=1}^{N}q^ir^i$ for $r=(r^0,\ldots,r^N)\in\mathbb{R}^{N+1}$.

For any $(y,r)\in\mathcal{R}$, we define the dual set $\widetilde{\mathcal{Y}}(y,r)$ as a proper extension of the auxiliary set $\widetilde{\mathcal{M}}$ by
\begin{equation}
\widetilde{\mathcal{Y}}(y,r)\triangleq \Big\{\Gamma\in\mathbb{L}_+^0(\Omega\times[0,T]):\mathbb{E}\Big[\int_0^T\tilde{c}_t\Gamma_tdt\Big]\leq xy+(-z,q)\cdot r\ \text{for all}\ \tilde{c}\in\widetilde{\mathcal{A}}(x,q,z)\ \text{and}\ (x,q,z)\in\mathcal{L}\Big\}.\nonumber
\end{equation}

The \textbf{Auxiliary Dual Utility Maximization Problem} to $(\ref{newpro})$ can now be formulated as
\begin{equation}\label{axdual}
\tilde{v}(y,r)\triangleq \inf_{\Gamma\in\widetilde{\mathcal{Y}}(y,r)}\mathbb{E}\Big[\int_0^TV(t,\Gamma_t)dt\Big],\ \ (y,r)\in\mathcal{R}.
\end{equation}

Our main result is stated as the following theorem on the existence of optimal solutions to the abstract optimization problems and the conjugate duality between two value functions.

\begin{theorem}\label{mainth}
Let Assumption $\ref{assumZ}$, $\ref{endowassum}$, $\ref{notrep}$, $\ref{ASSUV}$ and $\ref{asssss}$ hold. Moreover, let
$\tilde{u}(x,q,z)<\infty$ for some $(x,q,z)\in\mathcal{L}$. Then we have
\begin{itemize}
\item[(i)] The function $\tilde{u}(x,q,z)$ is $(-\infty,\infty)$-valued on $\mathcal{L}$ and $\tilde{v}(y,r)$ is $(-\infty,\infty)$-valued on $\mathcal{R}$. The conjugate duality of value functions $\tilde{u}$ and $\tilde{v}$ holds:
\begin{equation}
\begin{split}
\tilde{u}(x,q,z)&=\inf_{(y,r)\in\mathcal{R}}\Big(\tilde{v}(y,r)+xy+(-z,q)\cdot r\Big),\ \ \ (x,q,z)\in\mathcal{L},\\
\tilde{v}(y,r)&=\sup_{(x,q,z)\in\mathcal{L}}\Big(\tilde{u}(x,q,z)-xy-(-z,q)\cdot r\Big),\ \ \ (y,r)\in\mathcal{R}.\nonumber
\end{split}
\end{equation}
\item[(ii)] The optimal solution $\Gamma^{\ast}(y,r)$ to the problem $(\ref{axdual})$ exists and is unique for all $(y,r)\in\mathcal{R}$.
\item[(iii)] The optimal solution $\tilde{c}^{\ast}(x,q,z)$ to the problem $(\ref{newpro})$ exists and is unique for all $(x,q,z)\in\mathcal{L}$. Moreover, there exists a representation in the equivalent class such that $\tilde{c}^{\ast}_t(x,q,z)>0$, $\mathbb{P}$-a.s. for $t\in[0,T]$.
\item [(iv)] The superdifferential of $\tilde{u}$ maps $\mathcal{L}$ into $\mathcal{R}$, i.e., $\partial \tilde{u}(x,q,z)\subset \mathcal{R}$ for $(x,q,z)\in\mathcal{L}$. In addition, if $(y,r)\in\partial\tilde{u}(x,q,z)$, $\tilde{c}^{\ast}(x,q,z)$ and $\Gamma^{\ast}(y,r)$ are related by
\begin{equation}
\begin{split}
&\Gamma_t^{\ast}(y,r)=U'(t,\tilde{c}_t^{\ast}(x,q,z)),\ \ \text{or}\ \ \ \tilde{c}_t^{\ast}(x,q,z)=I(t,\Gamma_t^{\ast}(y,r)),\ \ t\in[0,T],\\
&\mathbb{E}\Big[\int_0^T\tilde{c}_t^{\ast}(x,q,z)\Gamma_t^{\ast}(y,r)dt\Big]=xy+(-z,q)\cdot r.\nonumber
\end{split}
\end{equation}
\item[(v)] If we restrict the choice of initial wealth $x$, initial holding $q$ and initial habit formation $z$ such that $(x,q,z)\in\bar{\mathcal{L}}\subset\mathcal{L}$, the optimal solution $c^{\ast}(x,q,z)$ to the primal utility maximization problem $(\ref{primalp})$ exists and is unique. In addition, we have that for any $t\in[0,T]$,
\begin{equation}
\begin{split}
\tilde{c}_t^{\ast}(x,q,z)&=c_t^{\ast}(x,q,z)-F(c^{\ast})_t(x,q,z),\\
\text{or}\ \ \ c_t^{\ast}(x,q,z)&=\tilde{c}^{\ast}_{t}(x,q,z)+\int_{0}^{t}\delta_{s}e^{\int_{s}^{t}(\delta_{v}-\alpha_{v})dv}\tilde{c}^{\ast}_{s}(x,q,z)ds+ze^{\int_0^t(\delta_v-\alpha_v)dv}.\nonumber
\end{split}
\end{equation}
\end{itemize}
\end{theorem}

\begin{remark}
For $(x,q,z)\in\bar{\mathcal{L}}\subset\mathcal{L}$, if the optimal solution $\Gamma^{\ast}(y,r)$ of the auxiliary dual problem $(\ref{axdual})$ lies in the auxiliary dual set $\widetilde{\mathcal{M}}$ defined in $(\ref{initialdualM})$, i.e., there exists a SCPS $Z^{0,\ast}(y,r)$ such that
\begin{equation}\label{optimaldecomp}
\Gamma_t^{\ast}(y,r)=Z^{0,\ast}_t(y,r)+\delta_t\mathbb{E}\Big[\int_t^T e^{\int_t^s(\delta_v-\alpha_v)dv} Z_s^{0,\ast}(y,r)ds\Big|\mathcal{F}_t\Big],\ \ \forall t\in[0,T],
\end{equation}
the optimal consumption can be explicitly expressed by this SCPS $Z^{\ast}(y,r)\in\mathcal{Z}^s$ that
\begin{equation}\label{yuanlai}
\begin{split}
c_t^{\ast}(x,q,z)=& ze^{\int_0^t(\delta_v-\alpha_v)dv}+I\Big(t, Z^{0,\ast}_t(y,r)+\delta_t\mathbb{E}\Big[\int_t^T e^{\int_t^s(\delta_v-\alpha_v)dv} Z_s^{0,\ast}(y,r)ds\Big|\mathcal{F}_t\Big]\Big)\\
+&\int_{0}^{t}\delta_{s}e^{\int_{s}^{t}(\delta_{v}-\alpha_{v})dv}I\Big(s, Z^{0,\ast}_s(y,r)+\delta_s\mathbb{E}\Big[\int_s^T e^{\int_s^l(\delta_v-\alpha_v)dv} Z_l^{0,\ast}(y,r)dl\Big|\mathcal{F}_s\Big]\Big)ds,\ t\in[0,T].
\end{split}
\end{equation}
The impact of transaction costs on the optimal consumption stream is hidden implicitly in the definition of the dual set $\widetilde{\mathcal{M}}$ in $(\ref{initialdualM})$ and the choice of the SCPS $Z^{0,\ast}(y,r)$ in the decomposition form $(\ref{optimaldecomp})$. For general problems, we can not conclude that the optimal consumption is monotone in terms of the transaction costs $\Lambda$. Moreover, we can also observe that the optimal consumption depends intricately on discounting factors $\alpha$ and $\delta$. For instance, if $\delta$ increases, i.e., the consumption history has more weights in the $F(c^{\ast})$, the first term of the right hand side of $(\ref{yuanlai})$ increases but the second term of $(\ref{yuanlai})$ decreases and it is unclear if the third term of $(\ref{yuanlai})$ is monotone or not.
\end{remark}

In general, the dual optimizer $\Gamma^{\ast}(y,r)$ may not be in the set $\widetilde{\mathcal{M}}$ and hence $\Gamma^{\ast}(y,r)$ does not necessarily have the decomposition form in $(\ref{optimaldecomp})$ with a SCPS $Z^{0,\ast}(y,r)$. However, the dual optimizer may still have a nice decomposition form. In particular, we can derive the following special example with explicit properties on the optimal strategy.

\begin{corollary}\label{cormaincor}
Let us consider the market models with constant discounting factors $\delta$ and $\alpha$ and the Logarithmic Utility function $U(t,x)=\log x$. For the dual optimizer $\Gamma^{\ast}(y,r)$, define the process
\begin{equation}
Y^{\ast}_t(y,r)\triangleq \Gamma^{\ast}_t(y,r)-\delta_t\mathbb{E}\Big[\int_t^T\Gamma^{\ast}_s(y,r)e^{\int_t^s(-\alpha_v)dv} ds\Big|\mathcal{F}_t\Big],\ \ \forall t\in[0,T].\nonumber
\end{equation}
If the process $(Y_t)_{t\in[0,T]}$ is a strictly positive martingale, we have
\begin{equation}\label{diffdecompG}
\Gamma_t^{\ast}(y,r)=Y^{\ast}_t(y,r)+\delta_t\mathbb{E}\Big[\int_t^T e^{\int_t^s(\delta_v-\alpha_v)dv} Y_s^{\ast}(y,r)ds\Big|\mathcal{F}_t\Big],\ \ \forall t\in[0,T],
\end{equation}
and the corresponding optimal consumption strategy is given explicitly by
\begin{equation}\label{exlogconsum}
c_t^{\ast}(x,q,z)=\left\{
\begin{array}{rl}
ze^{(\delta-\alpha)t}+\frac{\delta-\alpha}{Y_t^{\ast}(y,r)(\delta e^{(\delta-\alpha)(T-t)}-\alpha)}+\int_0^t\frac{\delta-\alpha}{Y_s^{\ast}(y,r)}\frac{\delta e^{(\delta-\alpha)(t-s)}}{(\delta e^{(\delta-\alpha)(T-s)}-\alpha)}ds,&\ \ \delta\neq\alpha,\\
z+\frac{1}{Y_t^{\ast}(y,r)}\frac{1}{1+\delta(T-t)}+\int_0^t\frac{1}{Y_s^{\ast}(y,r)}\frac{\delta}{1+\delta(T-s)}ds,&\ \ \delta=\alpha,
\end{array}\right.
\end{equation}
for $t\in[0,T]$, where $(y,r)$ satisfies $xy+(-z,q)\cdot r=T$. The corresponding optimal habit formation or standard of living process is
\begin{equation}\label{exlogconhab}
F(c^{\ast}(x,q,z))_t=\left\{
\begin{array}{rl}
ze^{(\delta-\alpha)t}+\int_0^t\frac{\delta-\alpha}{Y_s^{\ast}(y,r)}\frac{\delta e^{(\delta-\alpha)(t-s)}}{(\delta e^{(\delta-\alpha)(T-s)}-\alpha)}ds,&\ \ \delta\neq \alpha,\\
z+\int_0^t\frac{1}{Y_s^{\ast}(y,r)}\frac{\delta}{1+\delta(T-s)}ds,&\ \ \delta=\alpha.
\end{array}\right.\end{equation}
In addition, we have the following the properties:
\begin{itemize}
\item[(i)] If $\delta-\alpha>0$ or $\delta=\alpha>0$, the standard of living process $(F(c^{\ast})_t)_{t\in[0,T]}$ is an increasing process in terms of time $t$.
\item[(ii)] If $\alpha=0$ and $\delta$ or $T$ is sufficiently large, the optimal consumption strategy asymptotically behaves like $c_t^{\ast}(x,q,z)\approx ze^{\delta t}$, $t\in[0,T]$, hence it almost satisfies the consumption ratcheting constraint, i.e., the consumption process is increasing in terms of time $t$.
\item[(iii)] If $\delta-\alpha\geq 0$, the process $(c_t^{\ast}(x,q,z)Y_t^{0,\ast}(y,r))_{t\in[0,T]}$ is a submartingale. If $Y_0^{\ast}=1$, there exists a probability measure $\mathbb{Q}^{\ast}\sim\mathbb{P}$ such that the optimal consumption $(c_t^{\ast}(x,q,z))_{t\in[0,T]}$ is a submartingale under $\mathbb{Q}^{\ast}$ where $\frac{d\mathbb{Q}^{\ast}}{d\mathbb{P}}=Y_T^{\ast}(y,r)$. If $\delta=\alpha=0$, the process $(c_t^{\ast}(x,q,z)Y_t^{0,\ast}(y,r))_{t\in[0,T]}$ is a martingale.
\item[(iv)] If $Y_0^{\ast}(y,r)=1$, the optimal initial consumption amount is explicitly given by
\begin{equation}\label{initialconm}
c_0^{\ast}(x,q,z)=\left\{
\begin{array}{rl}
z+\frac{\delta-\alpha}{\delta e^{(\delta-\alpha)T}-\alpha},&\ \ \delta\neq\alpha,\\
z+\frac{1}{1+\delta T},&\ \ \delta=\alpha.
\end{array}\right.
\end{equation}
\end{itemize}
\end{corollary}

\begin{remark}
Again, if the dual optimizer happens in the set $\widetilde{\mathcal{M}}$, we must have $Y^{\ast}(y,r)=Z^0(y,r)$ for some SCPS $Z\in\mathcal{Z}^s$ which also implies that $Y_0^{\ast}(y,r)=1$. Therefore, assertions $(iii)$ and $(iv)$ of Corollary $\ref{cormaincor}$ obviously hold for the case $Y^{\ast}(y,r)=Z^0(y,r)$. Nevertheless, the dual optimizer $\Gamma^{\ast}(y,r)$ may have the decomposition form $(\ref{diffdecompG})$ for some martingale $Y^{\ast}(y,r)$ which is not a SCPS for the asset price process $S$ and transaction costs $\Lambda$.
\end{remark}

\section{Market Isomorphism in Models with Transaction Costs}\label{section4}
The following assumption on the non-negative optional discounting factors is mandated in this section.
\begin{assumption}\label{determdis}
The process $(\delta_t-\alpha_t)_{t\in[0,T]}$ is a deterministic function of time $t$.
\end{assumption}

Consider the same asset price process $S$ with the transaction costs $\Lambda$. The next theorem states an important observation that our utility maximization with habit formation is equivalent to a standard time-separable utility maximization on consumption in the modified market model with the change of num\'{e}raire and new random endowments.

\begin{proposition}\label{isotransc}
Let Assumption $\ref{determdis}$ hold. Given the initial wealth $x>0$ and initial habit $z\geq 0$, the original optimization problem $(\ref{primalp})$ is equivalent to the following time separable utility maximization problem on consumption under the change of num\'{e}raire:
\begin{equation}\label{numeraireG}
\hat{u}(x,q,z)\triangleq \sup_{\hat{c}\in\hat{\mathcal{A}}(x,q,z)}\mathbb{E}\Big[\int_0^TU(t,\frac{\hat{c}_t}{G_t})dt\Big].
\end{equation}
Here, we consider the same underlying asset price process $S$ and transaction costs $\Lambda$. The set $\hat{\mathcal{A}}(x,q,z)$ is the set of all $(x, N_T)$-financeable consumption processes with the new random endowments $N_T\triangleq q\cdot\mathcal{E}_T-z\int_0^Te^{\int_0^t(-\alpha_v)dv}G_tdt$ and the external auxiliary process
\begin{equation}\label{ohG}
G_t\triangleq 1+\delta_t\int_t^Te^{\int_t^s(\delta_v-\alpha_v)dv}ds,\ \ t\in[0,T].
\end{equation}
In particular, the optimal consumption $c^{\ast}(x,q,z)$ with habit formation for problem $(\ref{primalp})$ satisfies
\begin{equation}
c_t^{\ast}(x,q,z)=\frac{\hat{c}^{\ast}_{t}(x,q,z)}{G_t}+\int_{0}^{t}\delta_{s}e^{\int_{s}^{t}(\delta_{v}-\alpha_{v})dv}\frac{\hat{c}^{\ast}_{s}(x,q,z)}{G_s}ds+ze^{\int_0^t(\delta_v-\alpha_v)dv},\ t\in[0,T],\nonumber
\end{equation}
where $\hat{c}^{\ast}(x,q,z)$ is the optimal consumption in the utility maximization problem $(\ref{numeraireG})$ with the num\'{e}raire $G_t$ and new random endowments $N_T$.
\end{proposition}

In special cases, the num\'{e}raire process $(G_t)_{t\in[0,T]}$ may not affect the utility maximization problem in Proposition $\ref{isotransc}$ and hence the isomorphic optimization problem becomes a classical optimal consumption problem. We list below two examples to illustrate the simplicity of the isomorphic problems.

\begin{corollary}\label{exampleisomo1}
Let us consider the Logarithmic Utility function $U(t,x)=\log x$. Under the Assumption $\ref{determdis}$, the original utility maximization problem with habit formation $(\ref{primalp})$ is equivalent to the isomorphic utility maximization problem on consumption
\begin{equation}
\hat{u}(x,q,z)\triangleq \sup_{\hat{c}\in\hat{\mathcal{A}}(x,q,z)}\mathbb{E}\Big[\int_0^TU(t,\hat{c}_t)dt\Big]-\mathbb{E}\Big[\int_0^TU(t,G_t)dt\Big],\nonumber
\end{equation}
where $\hat{\mathcal{A}}(x,q,z)$ is defined the same as in Proposition $(\ref{isotransc})$. Therefore, it is enough to consider the standard utility maximization problem
\begin{equation}\label{OK1new}
\pi(x,q,z)=\sup_{\hat{c}\in\hat{\mathcal{A}}(x,q,z)}\mathbb{E}\Big[\int_0^TU(t,\hat{c}_t)dt\Big].
\end{equation}
Moreover, we have
\begin{equation}
c_t^{\ast}(x,q,z)=\frac{\hat{c}^{\ast}_{t}(x,q,z)}{G_t}+\int_{0}^{t}\delta_{s}e^{\int_{s}^{t}(\delta_{v}-\alpha_{v})dv}\frac{\hat{c}^{\ast}_{s}(x,q,z)}{G_s}ds+ze^{\int_0^t(\delta_v-\alpha_v)dv},\ t\in[0,T],\nonumber
\end{equation}
where $\hat{c}^{\ast}(x,q,z)$ is the optimal consumption of the problem $(\ref{OK1new})$ in the isomorphic market with random endowments $N_T=q\cdot\mathcal{E}_T-z\int_0^Te^{\int_0^t(-\alpha_v)dv}G_tdt$
\end{corollary}

\begin{corollary}\label{exampleisomo2}
If the num\'{e}raire process $(G_t)_{t\in[0,T]}$ is a martingale with $G_0=1$, under Assumption $\ref{determdis}$, the original utility maximization problem with habit formation $(\ref{primalp})$ is equivalent to the standard isomorphic utility maximization problem on consumption
\begin{equation}\label{OK2new}
u(x,q,z)=\sup_{\tilde{c}\in\bar{\mathcal{A}}(x,q,z)}\mathbb{E}\Big[\int_0^TU(t,\tilde{c}_t)dt\Big],
\end{equation}
where $\bar{\mathcal{A}}(x,q,z)$ defined in $(\ref{comebackA})$ is the set of all $(x,R_T)$-financeable consumption processes in the isomorphic market with the asset price process $(S_t)_{t\in[0,T]}$ and transaction costs $\Lambda$, where $R_T\triangleq q\cdot\mathcal{E}_T-z\int_0^Te^{\int_0^t(-\alpha_v)dv}dt$. Moreover, we have
\begin{equation}
c_t^{\ast}(x,q,z)=\tilde{c}^{\ast}_{t}(x,q,z)+\int_{0}^{t}\delta_{s}e^{\int_{s}^{t}(\delta_{v}-\alpha_{v})dv}\tilde{c}^{\ast}_{s}(x,q,z)ds+ze^{\int_0^t(\delta_v-\alpha_v)dv},\ t\in[0,T],\nonumber
\end{equation}
where $\tilde{c}^{\ast}(x,q,z)$ is the optimal consumption of the problem $(\ref{OK2new})$ in the isomorphic market with random endowments $R_T$.
\end{corollary}

The market isomorphism reduces the complexity of the path-dependence significantly as it is enough to find the optimal consumption in the new market model without habit formation constraint. It is worth noting that if there exists a  shadow price process $\hat{S}$ for the asset price process $S$ and transaction costs $\Lambda$, the market isomorphism can be carried out in the frictionless shadow price market with the underlying price process $\hat{S}$ which can even permit some closed-form feedback consumption strategies with habit formation and transaction costs.

\section{Proofs of Main Results}\label{section5}
\subsection{Proofs of Main Results in Section $\ref{section1}$}
\begin{proof}[Proof of Lemma $\ref{ran}$]
For each $\tilde{S}\in\mathcal{S}^s$, Assumption $\ref{endowassum}$ implies that
\begin{equation}
a\triangleq \sup_{\mathbb{Q}\in\mathcal{M}^s(\tilde{S})}\mathbb{E}^{\mathbb{Q}}\Big[\sum_{i=1}^{N}\mathcal{E}_T^i\Big]<\infty.\nonumber
\end{equation}
According to the general duality result between terminal wealth and equivalent local martingale measures in the frictionless market with the stock price $\tilde{S}$, see \cite{kram99}, it follows that there exists a nonnegative wealth process $(X_t)_{t\in[0,T]}=(a+(H\cdot\tilde{S})_t)_{t\in[0,T]}\in\mathcal{X}(\tilde{S},a)$ such that $X_T\geq \sum_{i=1}^{N}\mathcal{E}_T^i$. If $X$ itself is a maximal element in $\mathcal{X}(\tilde{S}, a)$, the conclusion holds obviously. Otherwise, as $X$ is not a maximal element, there exists some portfolio $\hat{H}$ such that $X_T\leq \hat{X}_T=a+(\hat{H}\cdot \tilde{S})_T$. Without loss of generality, we can assume $\hat{X}$ is the maximal element in $\mathcal{X}(\tilde{S}, a)$ and the conclusion still holds.
\end{proof}

\begin{proof}[Proof of Lemma $\ref{budlem}$]
By the definition of set $\mathcal{C}(x,q)$, for any $g\in\mathcal{C}(x,q)$, there exists a $V\in\mathcal{X}(x,q)$ such that $V_T+(q\cdot\mathcal{E}_T,\bar{\mathbf{0}})-g\in\mathbb{L}^0(\hat{K}_T)$, hence it is enough to show that
\begin{equation}\label{auxinqt}
\mathbb{E}[\langle V_T+(q\cdot\mathcal{E}_T, \bar{\mathbf{0}}), Z_T\rangle]\leq \langle x, Z_0\rangle+\mathbb{E}[\langle (q\cdot\mathcal{E}_T,\bar{\mathbf{0}}), Z_T\rangle],\ \ \forall Z\in\mathcal{Z}^s.
\end{equation}

By Remark~\ref{rmkCPS}, it is equivalent to prove that for each fixed $\tilde{S}\in\mathcal{S}^s$,
\begin{equation}\label{proineq}
\mathbb{E}^{\mathbb{Q}}\Big[V_T^0+\sum_{i=1}^{d}\tilde{S}^i_TV^i_T+q\cdot\mathcal{E}_T\Big]\leq \langle x, Z_0\rangle+\mathbb{E}^{\mathbb{Q}}[q\cdot\mathcal{E}_T],\ \ \forall \mathbb{Q}\in\mathcal{M}^s(\tilde{S}).
\end{equation}

The definition of acceptable portfolio implies that there exists $a>0$ and for each fixed $\tilde{S}\in\mathcal{S}^s$, there exists a maximal element $X^{\max,\tilde{S}}\in\mathcal{X}(\tilde{S},a)$ such that
\begin{equation}\label{proineq3}
\langle V_{\tau}+(X^{\max,\tilde{S}}_{\tau},\bar{\mathbf{0}}), Z_{\tau}\rangle\geq 0, \ \ \mathbb{P}-\text{a.s.},\ \forall Z\in\mathcal{Z}^s(\tilde{S})
\end{equation}
for all $[0,T]$-valued stopping time $\tau$. For this choice of $X^{\max,\tilde{S}}$, we can rewrite
\begin{equation}\label{proineq2}
\begin{split}
&\mathbb{E}^{\mathbb{Q}}\Big[V_T^0+\sum_{i=1}^{d}\tilde{S}_T^{i}V_T^{i}+q\cdot\mathcal{E}_T \Big]=\mathbb{E}^{\mathbb{Q}}\Big[V_T^0+\sum_{i=1}^{d}\tilde{S}_T^{i}V_T^{i}+X^{\max, \tilde{S}}_T-X^{\max,\tilde{S}}_T+q\cdot\mathcal{E}_T \Big]\\
=&\mathbb{E}^{\mathbb{Q}}\Big[V_T^0+\sum_{i=1}^{d}\tilde{S}_T^{i}V_T^{i}+X^{\max, \tilde{S}}_T\Big] -\mathbb{E}^{\mathbb{Q}}[X^{\max,\tilde{S}}_T]+\mathbb{E}^{\mathbb{Q}}[q\cdot\mathcal{E}_T ].
\end{split}
\end{equation}

Theorem $5.2$ of \cite{DS97} states that for each fixed semimartingale $\tilde{S}\in\mathcal{S}^s$, the set
\begin{equation}
\mathcal{M}'(\tilde{S})\triangleq \{\mathbb{Q}\in\mathcal{M}^s(\tilde{S}): X^{\max,\tilde{S}} \text{ is a UI martingale under } \mathbb{Q}\}\nonumber
\end{equation}
is nonempty and dense in $\mathcal{M}^s(\tilde{S})$ with respect to the norm topology of $\mathbb{L}^1(\Omega,\mathcal{F},\mathbb{P})$. Therefore, we can first prove the inequality $(\ref{proineq})$ for all $\mathbb{Q}\in\mathcal{M}'(\tilde{S})$ instead of $\mathbb{Q}\in\mathcal{M}^s(\tilde{S})$. By $(\ref{proineq2})$ and the fact that $X^{\max,\tilde{S}}$ is a UI martingale under $\mathbb{Q}\in\mathcal{M}'(\tilde{S})$, it is sufficient to show that
\begin{equation}\label{proineq4}
\mathbb{E}^{\mathbb{Q}}\Big[V_T^0+\sum_{i=1}^{d}\tilde{S}_T^{i}V_T^{i}+X^{\max, \tilde{S}}_T\Big]\leq \langle x, Z_0\rangle +a,\ \ \forall \mathbb{Q}\in\mathcal{M}'(\tilde{S}).
\end{equation}

However, since $V$ is a self-financing portfolio, Lemma $2.8$ of \cite{Campi06} implies that $\langle V, Z\rangle$ is a local supermartingale for all $Z\in\mathcal{Z}^s$, and hence $V_t^0+\sum_{i=1}^{d}\tilde{S}_t^{i}V_t^{i}$ is a local supermartingale under all $\mathbb{Q}\in\mathcal{M}'(\tilde{S})$. Again, since $X^{\max, \tilde{S}}$ is a UI martingale under $\mathbb{Q}\in\mathcal{M}'(\tilde{S})$, we obtain that $V_t^0+\sum_{i=1}^{d}\tilde{S}_t^{i}V_t^{i}+X^{\max, \tilde{S}}_t$ is also a local supermartingale under $\mathbb{Q}\in\mathcal{M}'(\tilde{S})$. Moreover, by $(\ref{proineq3})$, we have that $Z^0_{\tau}\Big(V_{\tau}^0+\sum_{i=1}^{d}\tilde{S}_{\tau}^{i}V_{\tau}^{i}+X^{\max, \tilde{S}}_{\tau}\Big)\geq 0$, $\mathbb{P}$-a.s. for all $[0,T]$-valued stopping time $\tau$ where $Z^0_T=\frac{d\mathbb{Q}}{d\mathbb{P}}$. It follows that $V_t^0+\sum_{i=1}^{d}\tilde{S}_t^{i}V_t^{i}+X^{\max, \tilde{S}}_t$ is a true supermartingale under $\mathbb{Q}\in\mathcal{M}'(\tilde{S})$ and $(\ref{proineq4})$ holds true.

Define $\beta_T\triangleq V_T^0+\sum_{i=1}^{d}\tilde{S}_T^iV_T^i+q\cdot\mathcal{E}_T$. Since $V_T+(q\cdot\mathcal{E}_T,\bar{\mathbf{0}})\in\mathbb{L}^0(\hat{K}_T)$, we get $\langle V_T+(q\cdot\mathcal{E}_T,\bar{\mathbf{0}}), Z_T\rangle\geq 0$ for all $Z\in\mathcal{Z}^s$. It thereby follows that for any $\mathbb{Q}\in\mathcal{M}^s(\tilde{S})$, we have $\beta_TZ^0_T\geq 0$, $\mathbb{P}$-a.s. where $Z_T^0=\frac{d\mathbb{Q}}{d\mathbb{P}}$. Monotone Convergence Theorem implies that for any $\mathbb{Q}\in\mathcal{M}^s(\tilde{S})$, the following holds
\begin{equation}
\mathbb{E}^{\mathbb{Q}}[\beta_T]=\mathbb{E}^{\mathbb{P}}[\beta_TZ_T^0]=\lim_{m\rightarrow\infty}\mathbb{E}^{\mathbb{P}}[\beta_T\mathbf{1}_{\{\beta_T\leq m\}}Z_T^0]=\lim_{m\rightarrow\infty}\mathbb{E}^{\mathbb{Q}}[\beta_T\mathbf{1}_{\{\beta_T\leq m\}}].\nonumber
\nonumber
\end{equation}
The density property of $\mathcal{M}'(\tilde{S})$ in $\mathcal{M}^s(\tilde{S})$ in the norm topology of $\mathbb{L}^1$ guarantees the existence of a sequence of $\mathbb{Q}^n\in\mathcal{M}'(\tilde{S})$ such that $(\ref{proineq4})$ holds. It follows that
\begin{equation}
\begin{split}
\lim_{m\rightarrow\infty}\mathbb{E}^{\mathbb{Q}}[\beta_T\mathbf{1}_{\{\beta_T\leq m\}}]&=\lim_{m\rightarrow\infty}\lim_{n\rightarrow\infty}\mathbb{E}^{\mathbb{Q}^n}[\beta_T\mathbf{1}_{\{\beta_T\leq m\}}]\\
&\leq \lim_{n\rightarrow\infty}\mathbb{E}^{\mathbb{Q}^n}[\beta_T]\leq \langle x,Z_0\rangle +\lim_{n\rightarrow\infty}\mathbb{E}^{\mathbb{Q}^n}[q\cdot\mathcal{E}_T].\nonumber
\end{split}
\end{equation}

It is clear that for any $m>0$ and each $1\leq i\leq N$, we have
\begin{equation}
\mathcal{E}^i_T\mathbf{1}_{\{\mathcal{E}^i_T>m\}}\leq \sum_{i=1}^{N}\mathcal{E}_T^i\mathbf{1}_{\{\sum_{i=1}^{N}\mathcal{E}_T^i >m\}},\ \ \mathbb{P}-\text{a.s.}.\nonumber
\end{equation}
The assumption that $\mathcal{E}^i_T\geq 0$ a.s. under $\mathbb{P}$ implies $\mathcal{E}^i_T\geq 0$ a.s. under $\mathbb{Q}\in\mathcal{M}^s(\tilde{S})$, it follows that
\begin{equation}\label{woo}
\lim_{m\rightarrow\infty}\sup_{\mathbb{Q}\in\mathcal{M}^s(\tilde{S})}\mathbb{E}^{\mathbb{Q}}[\mathcal{E}^i_T\mathbf{1}_{\{\mathcal{E}^i_T>m\}}]=0,\ \ 1\leq i\leq N.
\end{equation}

By $(\ref{woo})$ as well as the Moore-Osgood Theorem (see Theorem $5$, p.$102$ of \cite{13}) and Monotone Convergence Theorem, we deduce that
\begin{equation}
\begin{split}
\lim_{n\rightarrow\infty}\mathbb{E}^{\mathbb{Q}^n}[\mathcal{E}^i_T]=&\lim_{n\rightarrow\infty}\lim_{m\rightarrow\infty}\mathbb{E}^{\mathbb{Q}^n}[\mathcal{E}^i_T\mathbf{1}_{\{\mathcal{E}^i_T\leq m\}}]=\lim_{m\rightarrow\infty}\lim_{n\rightarrow\infty}\mathbb{E}^{\mathbb{Q}^n}[\mathcal{E}^i_T\mathbf{1}_{\{\mathcal{E}^i_T\leq m\}}]\\
=&\lim_{m\rightarrow\infty}\mathbb{E}^{\mathbb{Q}}[\mathcal{E}^i_T\mathbf{1}_{\{\mathcal{E}^i_T\leq m\}}]=\mathbb{E}^{\mathbb{Q}}[\mathcal{E}^i_T],\ \ 1\leq i\leq N, \nonumber
\end{split}
\end{equation}
which gives
\begin{equation}
\lim_{n\rightarrow\infty}\mathbb{E}^{\mathbb{Q}^n}[q\cdot\mathcal{E}_T]=\mathbb{E}^{\mathbb{Q}}[q\cdot\mathcal{E}_T].\nonumber
\end{equation}

As a consequence, $(\ref{proineq})$ holds for any $\mathbb{Q}\in\mathcal{M}^s(\tilde{S})$. Since $\tilde{S}\in\mathcal{S}^s$ is arbitrary, $(\ref{auxinqt})$ is verified which completes the proof of $(\ref{leminq})$.
\end{proof}

Fix a constant $\hat{a}>0$ and denote $\mathcal{V}_{0,\hat{a}}^{\text{acpt}}$ the set of all acceptable portfolios $V$ with zero initial wealth $V_0=(0,\bar{\mathbf{0}})$ and for each $\tilde{S}$, there exists a $\hat{X}^{\max,\tilde{S}}\in\mathcal{X}(\hat{a},\tilde{S})$ with $V_T+(\hat{X}_T^{\max,\tilde{S}}, \bar{\mathbf{0}})\in\mathbb{L}^0(\hat{K}_T)$. The following lemma asserts that the total variation of the element in $\mathcal{V}_{0,\hat{a}}^{\text{acpt}}$ remains bounded in probability.

\begin{lemma}\label{bdp}
Let Assumption $\ref{assumZ}$ hold. For each $\hat{a}>0$, there exists a probability measure $\mathbb{Q}\sim\mathbb{P}$ and a constant $C>0$ such that for all $V\in\mathcal{V}_{0,\hat{a}}^{\text{acpt}}$,
\begin{equation}
\mathbb{E}^{\mathbb{Q}}[\|V\|_T]\leq C\hat{a}.\nonumber
\end{equation}
\end{lemma}

\begin{proof}[Proof of Lemma $\ref{bdp}$]
For a fixed SCPS $Z\in\mathcal{Z}^s$, there exists a $\tilde{S}\in\mathcal{S}^s$ such that $Z\in\mathcal{Z}^s(\tilde{S})$. From the definition of acceptable portfolios, we can find a constant $\hat{a}>0$ and $X^{\max,\tilde{S}}\in\mathcal{X}(\tilde{S}, \hat{a})$ such that $V_T+(X^{\max,\tilde{S}}_T,\bar{\mathbf{0}})\in\mathbb{L}^0(\hat{K}_T)$ and $\langle V_{\tau}+(X^{\max,\tilde{S}}_{\tau},\bar{\mathbf{0}}), Z_{\tau}\rangle\geq 0$ a.s. for all $[0,T]$-valued stopping time $\tau$ and for all $Z\in\mathcal{Z}^s(\tilde{S})$.

Similar to Lemma $2.8$ of \cite{Campi06}, applying the integration by parts formula, we get
\begin{equation}
\langle V_t, Z_t\rangle=\int_0^tV_udZ_u+\int_0^tZ_u\dot{V}^c_ud\|V^c\|_u+\sum_{u\leq t}Z_{u-}\triangle V_u+\sum_{u<t}Z_u\triangle_{+}V_u,\nonumber
\end{equation}
where the first integral on the right hand side is a local martingale since $V$ is predictable with finite variation and hence locally bounded, see Proposition $A.11$ of \cite{Paolo1}. Thus there exists a sequence of stopping times $\tau_n\nearrow T$ such that $\int_0^{\tau_n}V_udZ_u$ is a martingale. Moreover, as $V$ is locally bounded, without loss of generality, we can assume that $|V_{\tau_n}|\leq n$. By taking the expectation, we arrive at
\begin{equation}
\mathbb{E}\bigg[ -\int_0^{\tau_n}Z_u\dot{V}^c_ud\|V^c\|_u-\sum_{u\leq {\tau_n}}Z_{u-}\triangle V_u-\sum_{u<{\tau_n}}Z_u\triangle_{+}V_u\bigg]=-\mathbb{E}[\langle V_{\tau_n}, Z_{\tau_n}\rangle].\nonumber
\end{equation}
The self-financing condition and Lemma $2.8$ of \cite{Campi06} imply that the process on the left hand side is non-negative. Fatou's Lemma gives that
\begin{equation}
\mathbb{E}\bigg[ -\int_0^{T}Z_u\dot{V}^c_ud\|V^c\|_u-\sum_{u\leq {T}}Z_{u-}\triangle V_u-\sum_{u<{T}}Z_u\triangle_{+}V_u\bigg]\leq -\mathbb{E}[\langle V_{\tau_n}, Z_{\tau_n}\rangle],\ \ \forall n\in\mathbb{N}.\nonumber
\end{equation}

For the right hand side, it is clear from the tower property and martingale property of $Z$ that
\begin{equation}
-\mathbb{E}[\langle V_{\tau_n}, Z_{\tau_n}\rangle]=-\mathbb{E}[\langle V_{\tau_n}, Z_T\rangle].\nonumber
\end{equation}

For the fixed $Z\in\mathcal{Z}(\tilde{S})$, $\tilde{S}\in\mathcal{S}^s$ and the stochastic lower bound $X^{\max, \tilde{S}}$ in the definition of $V$, we consider the set of measures
\begin{equation}
\mathcal{M}'(\tilde{S})\triangleq \{\mathbb{Q}\in\mathcal{M}^s(\tilde{S}): X^{\max,\tilde{S}}\ \text{is a UI martingale under}\ \mathbb{Q}\}.\nonumber
\end{equation}
The density property of $\mathcal{M}'(\tilde{S})$ in $\mathcal{M}^s(\tilde{S})$ implies the existence of a sequence of $\mathbb{Q}^m\in\mathcal{M}'(\tilde{S})$ such that $\mathbb{Q}^m$ converges to $\mathbb{Q}$ in the norm topology of $\mathbb{L}^1$ where $\frac{d\mathbb{Q}}{d\mathbb{P}}=Z_T$. Let us define the sequence $Z^m$ by
\begin{equation}
\begin{split}
Z^{m,0}_t&=\mathbb{E}\Big[\frac{d\mathbb{Q}^m}{d\mathbb{P}}\Big|\mathcal{F}_t\Big],\\
Z^{m,i}_t&=Z^i_t,\ \ i=1,\ldots,d,\ \ t\in[0,T].\nonumber
\end{split}
\end{equation}

As $|V_{\tau_n}|\leq n$, it is clear that
\begin{equation}
-\mathbb{E}[\langle V_{\tau_n}, Z_T\rangle]=-\lim_{m\rightarrow\infty}\mathbb{E}[\langle V_{\tau_n}, Z^m_T\rangle]=-\lim_{m\rightarrow\infty}\mathbb{E}[\langle V_{\tau_n}, Z^m_{\tau_n}\rangle].\nonumber
\end{equation}

Following the proof of Lemma $\ref{budlem}$, we obtain that $\langle V, Z^m\rangle$ is a true supermartingale since $\langle V, Z^m\rangle$ is a local supermartingale by Lemma $2.8$ of \cite{Campi06} and it is also bounded below by the UI martingale $\langle (X^{\max,\tilde{S}}, \bar{\mathbf{0}}), Z^m\rangle$. It follows that
\begin{equation}
\begin{split}
-\lim_{m\rightarrow\infty}\mathbb{E}[\langle V_{\tau_n}, Z^m_{\tau_n}\rangle]&\leq -\lim_{m\rightarrow\infty}\mathbb{E}[\langle V_{T}, Z^m_{T}\rangle]\leq\lim_{m\rightarrow\infty} \mathbb{E}[\langle  (\hat{X}_T^{\max,\tilde{S}}, \bar{\mathbf{0}}), Z^m \rangle]\\
&=\lim_{m\rightarrow\infty}\mathbb{E}[\hat{X}_T^{\max,\tilde{S}}Z_T^{m,0}]=\hat{a}.\nonumber
\end{split}
\end{equation}

Putting all pieces together, we can deduce that
\begin{equation}\label{newleftside}
\mathbb{E}\bigg[ -\int_0^{T}Z_u\dot{V}^c_ud\|V^c\|_u-\sum_{u\leq {T}}Z_{u-}\triangle V_u-\sum_{u<{T}}Z_u\triangle_{+}V_u\bigg]\leq \hat{a}.
\end{equation}

For the fixed $Z\in\mathcal{Z}^s$, let us define the random variable $\alpha(Z)\triangleq \epsilon(Z)\inf_{t\in[0,T]}|Z_t|_{1+d}$ where
\begin{equation}\label{epsi}
\epsilon(Z)\triangleq \text{esssup}\big\{\eta\in\mathbb{L}^0(\mathbb{R}_+,\mathcal{F}_T): Z_t\in\eta\text{-int}\hat{K}_t^{\ast},\ \ \forall t\in[0,T]\ \text{a.s.}\big\}.
\end{equation}
Lemma $3.1$ of \cite{Campi06} states that $\mathbb{P}(\epsilon(Z)>0)=1$ and also $\inf_{t\in[0,T]}|Z_t|_{1+d}>0$ a.s.. The inequality $(\ref{newleftside})$ implies that
\begin{equation}
\mathbb{E}[\alpha(Z)\|V\|_T]\leq \hat{a}.\nonumber
\end{equation}

Similar to the proof of Lemma $3.2$ of \cite{Campi06}, for the fixed $Z\in\mathcal{Z}^s$, we can then define $C\triangleq \mathbb{E}[\alpha(Z)]$ and $\frac{d\mathbb{Q}}{d\mathbb{P}}\triangleq \frac{\alpha(Z)}{C}$, and the conclusion holds.
\end{proof}

Denote $\mathcal{A}_x=\{V_T: V\in\mathcal{V}_x^{\text{acpt}}\}$ the set of all contingent claims attainable with initial position $x\in\mathbb{R}^{1+d}$. We will modify the closedness under Fatou convergence to fit into our framework in the following way:
\begin{definition}\label{newfffttto}
The set $\mathcal{A}_x$ is said to be \textbf{relatively Fatou closed} if for any fixed $\hat{a}>0$ and for each $\tilde{S}\in\mathcal{S}^s$ with one fixed maximal element $\hat{X}^{\max,\tilde{S}}\in\mathcal{X}(\tilde{S},\hat{a})$ such that there exists a sequence $V^n\in\mathcal{V}_x^{\text{acpt}}$ satisfying $V_T^n+(\hat{X}_T^{\max,\tilde{S}}, \bar{\mathbf{0}})\in\mathbb{L}^0(\hat{K}_T)$ for each $\tilde{S}\in\mathcal{S}^s$ converges almost surely to a $\mathcal{F}_T$-measurable random variable $V_T$, we have that $V_T\in\mathcal{A}_x$ holds.
\end{definition}

\begin{remark}
The new terminology we choose here comes from the fact that the converging sequence $(V^n)_{n\in\mathbb{N}}$ has to satisfy the lower bound condition $V_T^n+(\hat{X}_T^{\max,\tilde{S}}, \bar{\mathbf{0}})\in\mathbb{L}^0(\hat{K}_T)$ relatively to the fixed maximal elements $\hat{X}^{\max,\tilde{S}}$ for each $\tilde{S}$. The definition of relatively Fatou closed is more restrictive than the Fatou closedness. If a set is relatively Fatou closed, then it is obviously Fatou closed since each constant $\hat{a}>0$ is one maximal element in the set $\mathcal{X}(\tilde{S}, \hat{a})$ for all $\tilde{S}\in\mathcal{S}^s$. We need this stronger condition on the set due to the complexity caused by the definition of acceptable portfolios. The relatively Fatou closedness in the end will help us to deduce the characterization of the set $\mathcal{A}_x$ using some explicit dual elements.
\end{remark}

\begin{lemma}\label{fcl}
Under Assumption $\ref{assumZ}$, the set $\mathcal{A}_x$ is relatively Fatou closed.
\end{lemma}

\begin{proof}[Proof of Lemma $\ref{fcl}$]
The conclusion holds if and only if we can show $\mathcal{A}_0$ is relatively Fatou closed since it is easy to see that $V\in\mathcal{V}_0^{\text{acpt}}$ if and only if $V+x\in\mathcal{V}_x^{\text{acpt}}$. Therefore, we will only prove that $\mathcal{A}_0$ is relatively Fatou closed.

Given $\hat{a}>0$ and for each $\tilde{S}\in\mathcal{S}^s$, choose and fix a maximal element $\hat{X}^{\max,\tilde{S}}\in\mathcal{X}(\tilde{S},\hat{a})$. Let $V^n$ be a sequence in $\mathcal{V}_0^{\text{acpt}}$ satisfying $V_T^n+(\hat{X}_T^{\max,\tilde{S}}, \bar{\mathbf{0}})\in\mathbb{L}^0(\hat{K}_T)$ for all $\tilde{S}\in\mathcal{S}^s$, so we get $V^n\in\mathcal{V}_{0,\hat{a}}^{\text{acpt}}$. Assume that $V_T^n$ converges a.s. to some $\mathcal{F}_T$ measurable random variable $X\in\mathbb{L}^0(\mathbb{R}^{1+d})$. According to Lemma $\ref{bdp}$ and Proposition $3.4$ of \cite{Campi06}, there exists a sequence of convex combinations of $V^n$ still denoted by $V^n$ which converges to some finite variation, predictable process $V$ pointwise. Hence we immediately get $V_T=X$ and $V_T+(\hat{X}_T^{\max,\tilde{S}},\bar{\mathbf{0}})\in\mathbb{L}^0(\hat{K}_T)$ for all $\tilde{S}\in\mathcal{S}^s$. It is easy to check that condition $(\ref{selfcond})$ holds and hence $V$ is a self-financing portfolio process. It is enough to prove that it is an acceptable portfolio process. For each $\tilde{S}\in\mathcal{S}^s$, we consider any $\hat{Z}\in\mathcal{Z}^s(\tilde{S})$ and the maximal element $\hat{X}^{\max,\tilde{S}}$, for any $t\in[0,T]$, we have
\begin{equation}
\begin{split}
&\langle V_t^n+(\hat{X}_t^{\max,\tilde{S}}, \bar{\mathbf{0}}), \hat{Z}_t\rangle=\Big(V_t^{0,n}+\sum_{i=1}^{d}\tilde{S}_t^iV_t^{i,n}+\hat{X}_t^{\max,\tilde{S}}\Big)\hat{Z}_t^0,\\
\text{and}\ \ \ &\langle V_t+(\hat{X}_t^{\max,\tilde{S}}, \bar{\mathbf{0}}), \hat{Z}_t\rangle=\Big(V_t^{0}+\sum_{i=1}^{d}\tilde{S}_t^iV_t^{i}+\hat{X}_t^{\max,\tilde{S}}\Big)\hat{Z}_t^0,\nonumber
\end{split}
\end{equation}
where we have $\hat{Z}_t^0=\mathbb{E}\Big[\frac{d\hat{\mathbb{Q}}}{d\mathbb{P}}\Big|\mathcal{F}_t\Big]$ and $\hat{\mathbb{Q}}\in\mathcal{M}^s(\tilde{S})$.

For the same $\tilde{S}\in\mathcal{S}^s$, we now consider the maximal element $X^{\max,\tilde{S}, n}$ in the definition of each acceptable portfolio $V^n$ as the lower bounded. Define the set
\begin{equation}
\mathcal{M}'(\tilde{S}, n)\triangleq \{\mathbb{Q}\in\mathcal{M}^s(\tilde{S}): X^{\max,\tilde{S}, n} \text{ is a UI martingale under } \mathbb{Q}\}\nonumber.
\end{equation}
It follows that there exists a sequence of $\mathbb{Q}^m\in\mathcal{M}'(\tilde{S}, n)$ converging to $\hat{\mathbb{Q}}$ as $m\rightarrow\infty$ in the norm topology of $\mathbb{L}^1$ under $\mathbb{P}$. By passing to the subsequence if necessary, we deduce that $Z_t^{0,m}$ converges to $\hat{Z}_t^{0}$ $\mathbb{P}$-a.s. for all $t\in[0,T]$ where $Z_t^{0,m}=\mathbb{E}[\frac{d\mathbb{Q}^m}{d\mathbb{P}}|\mathcal{F}_t]$. Therefore, for any $t\in[0,T]$, the following holds:
\begin{equation}
\Big(V_t^{0,n}+\sum_{i=1}^{d}\tilde{S}_t^iV_t^{i,n}+\hat{X}_t^{\max,\tilde{S}}\Big)\hat{Z}_t^0=\lim_{m\rightarrow\infty}\Big(V_t^{0,n}+\sum_{i=1}^{d}\tilde{S}_t^iV_t^{i,n}+\hat{X}_t^{\max,\tilde{S}}\Big)Z_t^{0,m}.\nonumber
\end{equation}
On the other hand, for each fixed $\mathbb{Q}\in\mathcal{M}'(\tilde{S},n)$, similar to the proof of Lemma $\ref{bdp}$, we obtain that $(V_t^{0,n}+\sum_{i=1}^{d}\tilde{S}_t^{i}V_t^{i,n})Z_t^{0,m}$ is a true supermartingale. Also, we have the stochastic integral $\hat{X}^{\max,\tilde{S}}$ is a supermartingale under $\mathbb{Q}^m$, hence $\hat{X}^{\max, \tilde{S}}Z^{0,m}$ is another supermartingale. It follows that for any $[0,T]$-valued stopping time $\tau$,
\begin{equation}
\Big(V_{\tau}^{0,n}+\sum_{i=1}^{d}\tilde{S}_{\tau}^iV_{\tau}^{i,n}+\hat{X}_{\tau}^{\max,\tilde{S}}\Big)Z_{\tau}^{0,m}\geq \mathbb{E}\Big[\Big(V_T^{0,n}+\sum_{i=1}^{d}\tilde{S}_T^iV_T^{i,n}+\hat{X}_T^{\max,\tilde{S}}\Big)Z_T^{0,m}\Big|\mathcal{F}_{\tau}\Big].\nonumber
\end{equation}
We know that $V_T^n+(\hat{X}_T^{\max,\tilde{S}},\bar{\mathbf{0}})\in\mathbb{L}^0(\hat{K}_T)$ which yields that  $(V_T^{0,n}+\sum_{i=1}^{d}\tilde{S}_T^iV_T^{i,n}+\hat{X}_T^{\max,\tilde{S}})Z_T^{0,m}\geq 0$, $\mathbb{P}$-a.s., and therefore Fatou's Lemma leads to
\begin{equation}
\Big(V_{\tau}^{0,n}+\sum_{i=1}^{d}\tilde{S}_{\tau}^iV_{\tau}^{i,n}+\hat{X}_{\tau}^{\max,\tilde{S}}\Big)\hat{Z}_{\tau}^{0}\geq \mathbb{E}\Big[\Big(V_T^{0,n}+\sum_{i=1}^{d}\tilde{S}_T^iV_T^{i,n}+\hat{X}_T^{\max,\tilde{S}}\Big)\hat{Z}_T^{0}\Big|\mathcal{F}_{\tau}\Big].\nonumber
\end{equation}
Since $V^n$ also converges to $V$ pointwise, again by $V_T^n+(\hat{X}_T^{\max,\tilde{S}},\bar{\mathbf{0}})\in\mathbb{L}^0(\hat{K}_T)$ and Fatou's Lemma, we obtain that
\begin{equation}
\Big(V_{\tau}^{0}+\sum_{i=1}^{d}\tilde{S}_{\tau}^iV_{\tau}^{i}+\hat{X}_{\tau}^{\max,\tilde{S}}\Big)\hat{Z}_{\tau}^{0}\geq \mathbb{E}\Big[\Big(V_T^{0}+\sum_{i=1}^{d}\tilde{S}_T^iV_T^{i}+\hat{X}_T^{\max,\tilde{S}}\Big)\hat{Z}_T^{0}\Big|\mathcal{F}_{\tau}\Big].\nonumber
\end{equation}
Therefore, we can see that $(V_{\tau}^{0}+\sum_{i=1}^{d}\tilde{S}_{\tau}^iV_{\tau}^{i}+\hat{X}_{\tau}^{\max,\tilde{S}})\hat{Z}_{\tau}^{0}\geq 0$, $\mathbb{P}$-a.s. which is equivalent to $\langle V_{\tau}+(\hat{X}^{\max,\tilde{S}}_{\tau},\bar{\mathbf{0}}), \hat{Z}_{\tau}\rangle\geq 0$, $\mathbb{P}$-a.s. for all $[0,T]$-valued stopping time $\tau$ and any $\hat{Z}\in\mathcal{Z}^s(\tilde{S})$. In conclusion, it is proved that the limit process $V$ is an acceptable portfolio with the constant $\hat{a}>0$ and $\hat{X}^{\max,\tilde{S}}\in\mathcal{X}(\tilde{S},\hat{a})$ which verifies the fact that $\mathcal{A}_0$ is relatively Fatou closed.
\end{proof}

\begin{lemma}
For each fixed $n\in\mathbb{N}$, let us define the set of truncated terminal liquidation values $\mathcal{A}_x^n\triangleq \{Y: Y=V_T\mathbf{1}_{\{|V_T|\leq n\}},\ V_T\in\mathcal{A}_x\}$. In addition, we consider the set $\mathcal{A}_x^{\infty}\triangleq \bigcup_{n\in\mathbb{N}}\mathcal{A}_x^n$. We have the following characterization
\begin{equation}\label{chaA}
\mathcal{A}_x^{\infty}=\Big\{Y\in\mathbb{L}_r^{\infty}(\mathbb{R}^{1+d},\mathcal{F}_T): \mathbb{E}[\langle Y, \eta\rangle]\leq \sup_{V_T\in\mathcal{A}_x^{\infty}}\mathbb{E}[\langle V_T, \eta\rangle],\ \forall \eta\in\mathbb{L}^1(\hat{K}^{\ast}_T)\Big\},
\end{equation}
where $\mathbb{L}_{r}^{\infty}(\mathbb{R}^{1+d},\mathcal{F}_T)$ is the set of all $\mathbb{R}^{1+d}$-valued and $\mathcal{F}_T$-measurable random vectors $Y\in\mathbb{L}^{\infty}$ such that there exists $a>0$ and for each $\tilde{S}$, there exists an $X^{\max,\tilde{S}}\in\mathcal{X}(\tilde{S}, a)$ with $Y+(X_T^{\max,\tilde{S}}, \bar{\mathbf{0}})\in\mathbb{L}^0(\hat{K}_T)$.
\end{lemma}
\begin{proof}
For any constant $\kappa>0$ and for each $\xi\in \{\xi:\|\xi\|_{\infty}\leq \kappa\}$, we can always find the constant $a=\kappa+1$ and for each $\tilde{S}\in\mathcal{S}^s$, we can choose $X^{\max,\tilde{S}}_t\equiv a\in\mathcal{X}(\tilde{S}, a)$ for all $t\in[0,T]$ so that $\xi+(X_T^{\max,\tilde{S}},\bar{\mathbf{0}})\subseteq\mathbb{L}^0(\mathbb{R}_+^{1+d})\subseteq\mathbb{L}^0(\hat{K}_T)$. In addition, for any $V_T\in\mathcal{A}_x$ and the corresponding $X^{\max,\tilde{S}}$, we claim that $V_T\mathbf{1}_{\{|V_T|\leq n\}}+(X_T^{\max,\tilde{S}},\bar{\mathbf{0}})\in\mathbb{L}^0(\hat{K}_T)$. To see this, it is noted that
\begin{equation}\label{truncV}
\begin{split}
V_T\mathbf{1}_{\{|V_T|\leq n\}}+(X_T^{\max,\tilde{S}}, \bar{\mathbf{0}})&\geq V_T\mathbf{1}_{\{|V_T|\leq n\}}+(X_T^{\max,\tilde{S}}, \bar{\mathbf{0}})\mathbf{1}_{\{|V_T|\leq n\}}\\
&=\Big(V_T+(X_T^{\max,\tilde{S}}, \bar{\mathbf{0}})\Big)\mathbf{1}_{\{|V_T|\leq n\}}\in\mathbb{L}^0(\hat{K}_T).
\end{split}
\end{equation}
Therefore, a bounded sequence convergent a.s. is also relatively Fatou convergent. Lemma $\ref{fcl}$ states that $\mathcal{A}_x$ is relatively Fatou closed, and thus it is straightforward to derive that an intersection of $\mathcal{A}_x^{\infty}$ with ball $\{\xi:\|\xi\|_{\infty}\leq \kappa\}$ is closed in probability for every $\kappa>0$. By the classical result, see Proposition $5.5.1$ of \cite{KS09}, $\mathcal{A}_x^{\infty}$ is weak$^{\ast}$ closed (i.e., closed in $\sigma(\mathbb{L}^{\infty}, \mathbb{L}^1)$). Following the same argument of Theorem $5.5.3$ of \cite{KS09}, we obtain that $(\ref{chaA})$ holds true.
\end{proof}

\begin{proof}[Proof of Lemma $\ref{lemsp}$]
It has been proved that $\mathcal{A}_x$ is relatively Fatou closed, and we now proceed with the verification that $\mathcal{A}_x^{\infty}$ is relatively Fatou-dense in $\mathcal{A}_x$. To this end, let us consider any $V_T\in\mathcal{A}_x$ with the constant $\hat{a}>0$ and $\hat{X}^{\max,\tilde{S}}\in\mathcal{X}(\tilde{S}, \hat{a})$ for each $\tilde{S}\in\mathcal{S}^s$. We need to show the existence of a sequence $V_T^n\in\mathcal{A}_x^{\infty}$ satisfying $V_T^n+(\hat{X}_T^{\max,\tilde{S}}, \bar{\mathbf{0}})\in\mathbb{L}^0(\hat{K}_T)$ for all $\tilde{S}\in\mathcal{S}^s$ as well as $V_T^n\rightarrow V_T$ a.s.. Fix $V_T$ and similar to Theorem $4.1$ of \cite{Campi06}, we consider the sequence of $V_T^n$ defined by
\begin{equation}
V_T^n=V_T\mathbf{1}_{\{|V_T|\leq n\}}.\nonumber
\end{equation}
It is clear that $V^n_T\in\mathcal{A}_x^n\subset \mathcal{A}_x^{\infty}$. Following the same argument in $(\ref{truncV})$, it is also easy to see that $V_T^n+(\hat{X}_T^{\max,\tilde{S}}, \bar{\mathbf{0}})\in\mathbb{L}^0(\hat{K}_T)$ for all $\tilde{S}\in\mathcal{S}^s$ for each $n\in\mathbb{N}$. Moreover,
\begin{equation}
V_T-\mathbb{L}^{\infty} \subseteq \mathcal{A}_x^{\infty},\nonumber
\end{equation}
for all $V_T\in\mathcal{A}_x^{\infty}$.

As $\mathcal{A}_x^{\infty}$ is relatively Fatou dense in $\mathcal{A}_x$, we can try to characterize elements in $\mathcal{A}_x$ using $(\ref{chaA})$. For any $Y\in \mathcal{A}_x$, we can find a constant $a>0$ and for each $\tilde{S}$ there exists an $X^{\max,\tilde{S}}\in\mathcal{X}(\tilde{S}, a)$ such that $Y+(X_T^{\max,\tilde{S}}, \bar{\mathbf{0}})\in\mathbb{L}^0(\hat{K}_T)$. It follows that there exists a sequence $Y^n\in\mathcal{A}_x^{\infty}$ such that $Y^n+(X_T^{\max,\tilde{S}}, \bar{\mathbf{0}})\in\mathbb{L}^0(\hat{K}_T)$ and $Y^n$ converges to $Y$ a.s.. Therefore, Fatou's lemma leads to
\begin{equation}\label{1chA}
\mathbb{E}[\langle Y+(X_T^{\max,\tilde{S}}, \bar{\mathbf{0}}), \eta\rangle]\leq \sup_{V_T\in\mathcal{A}_x^{\infty}}\mathbb{E}[\langle V_T, \eta\rangle] +\mathbb{E}[\langle (X_T^{\max,\tilde{S}}, \bar{\mathbf{0}}), \eta\rangle] ,\ \forall \eta\in\mathbb{L}^1(\hat{K}^{\ast}_T).
\end{equation}
The second expectation on the right hand side is well defined since $X_T^{\max,\tilde{S}}\geq 0$ a.s.. On the other hand, suppose that $Y+(X_T^{\max,\tilde{S}}, \bar{\mathbf{0}})\in\mathbb{L}^0(\hat{K}_T)$, we can construct $Y^n=Y\mathbf{1}_{\{|Y|\leq n\}} \in\mathbb{L}^{\infty}$ such that $Y^n$ relatively Fatou converges to $Y$. Moreover, by $(\ref{1chA})$, we can deduce that each $Y^n$ satisfies $\mathbb{E}[\langle Y^n, \eta\rangle]\leq  \sup_{V_T\in\mathcal{A}_x^{\infty}}\mathbb{E}[\langle V_T, \eta\rangle] $ for all $\eta\in\mathbb{L}^1(\hat{K}^{\ast}_T,\mathcal{F}_T)$. So $Y^n\in\mathcal{A}_x^{\infty}\subseteq\mathcal{A}_x$. The fact that $\mathcal{A}_x$ is relatively Fatou closed yields that $Y\in\mathcal{A}_x$.

It follows from the above argument that the set $\mathcal{A}_x$ can be rewritten as
\begin{equation}\label{2chA}
\begin{split}
\mathcal{A}_x=\Big\{&Y\in\mathbb{L}^0(\mathbb{R}^{1+d},\mathcal{F}_T):\ \text{there exists $a>0$ and for each $\tilde{S}\in\mathcal{S}^s$, there exists an}\\
&X^{\max,\tilde{S}}\in\mathcal{X}(\tilde{S},a)\ \text{with $Y+(X_T^{\max,\tilde{S}}, \bar{\mathbf{0}})\in\mathbb{L}^0(\hat{K}_T)$ and}\ \mathbb{E}[\langle Y+(X_T^{\max,\tilde{S}},\bar{\mathbf{0}}), \eta\rangle]\\
&\leq\sup_{V_T\in\mathcal{A}_x^{\infty}}\mathbb{E}[\langle V_T, \eta\rangle] +\mathbb{E}[\langle (X_T^{\max,\tilde{S}}, \bar{\mathbf{0}}), \eta\rangle] ,\ \forall \eta\in\mathbb{L}^1(\hat{K}^{\ast}_T)\Big\}.
\end{split}
\end{equation}

Pick $X\notin \mathcal{A}_x$. According to $(\ref{2chA})$, for any $a>0$ and $X^{\max,\tilde{S}}\in\mathcal{X}(\tilde{S},a)$, there exists $\eta\in\mathbb{L}^1(\hat{K}^{\ast}_T)$ such that
\begin{equation}\label{conch}
\mathbb{E}[\langle X+(X_T^{\max,\tilde{S}},\bar{\mathbf{0}}), \eta\rangle]> \sup_{V_T\in\mathcal{A}_x^{\infty}}\mathbb{E}[\langle V_T, \eta\rangle] +\mathbb{E}[\langle (X_T^{\max,\tilde{S}}, \bar{\mathbf{0}}), \eta\rangle].
\end{equation}
In particular, we can consider $X_t^{\max, \tilde{S}}\equiv a$ for $t\in[0,T]$. Therefore, $(\ref{conch})$ simplifies to be
\begin{equation}\label{conch2}
\mathbb{E}[\langle X, \eta\rangle]> \sup_{V_T\in\mathcal{A}_x^{\infty}}\mathbb{E}[\langle V_T, \eta\rangle].
\end{equation}
We can thus define the process $Z_t=\mathbb{E}[\eta|\mathcal{F}_t]$, $t\in[0,t]$ and obtain $\mathbb{E}[\langle X,Z_T\rangle]>\langle x,Z_0\rangle$ by its definition. Following the same proof of Theorem $4.1$ of \cite{Campi06}, it is easy to verify that $Z_t$ is a CPS. Let $Z^s$ be a SCPS, and for $0\leq \beta<1$ sufficiently small, the process $Z_t^{\beta}=\beta Z_t^s+Z_t$ is a SCPS. We will have $\mathbb{E}[\langle X, Z_T^{\beta}\rangle]> \langle x, Z_0\rangle $. According to $(\ref{superhedging_ineq})$, it follows that $g\in\mathcal{A}_x$. Therefore there exists a $V\in\mathcal{V}_x^{\text{acpt}}$ such that $V_T-g\in\mathbb{L}^0(\hat{K}_T)$.
\end{proof}

\subsection{Proofs of Main Results in Section $\ref{section2}$}

\begin{proof}[Proof of Proposition $\ref{bcccc}$]
If $c\in\mathcal{C}_{x,q\cdot\mathcal{E}_T}$, there exists an acceptable portfolio $V\in\mathcal{H}(x,q)$ such that
\begin{equation}
\Big\langle \Big(\int_0^Tc_tdt,\bar{\mathbf{0}}\Big), Z_T\Big\rangle\leq \langle V_T+(q\cdot\mathcal{E}_T, \bar{\mathbf{0}}), Z_T\rangle,\ \ \ \forall Z\in\mathcal{Z}^s. \nonumber
\end{equation}
Also, for each $Z\in\mathcal{Z}^s$, we have $\mathbb{E}\Big[\Big\langle \Big(\int_0^Tc_tdt,\bar{\mathbf{0}}\Big), Z_T\Big\rangle\Big]=\mathbb{E}\Big[\int_0^Tc_tdtZ^0_T\Big]$. Using integration by parts and choosing the localization sequence, we deduce that
\begin{equation}
\mathbb{E}\Big[\Big\langle \Big(\int_0^Tc_tdt,\bar{\mathbf{0}}\Big), Z_T\Big\rangle\Big]=\mathbb{E}\Big[\int_0^Tc_tZ_t^0dt\Big].\nonumber
\end{equation}
Following the proof of Lemma $\ref{budlem}$, we obtain that
\begin{equation}
\mathbb{E}[\langle V_T+(q\cdot\mathcal{E}_T, \bar{\mathbf{0}}), Z_T\rangle]\leq x+\mathbb{E}[\langle (q\cdot\mathcal{E}_T, \bar{\mathbf{0}}), Z_T\rangle],\ \ \ \forall Z\in\mathcal{Z}^s,\nonumber
\end{equation}
and hence $(\ref{budccc})$ holds.

On the other hand, for the process $c\geq 0$ which satisfies $(\ref{budccc})$, let us define $g\triangleq (\int_0^Tc_tdt-q\cdot\mathcal{E}_T, \bar{\mathbf{0}})$. By Lemma $\ref{ran}$, there exists a constant $\hat{a}>0$ and for each $\tilde{S}$, there exists an $\hat{X}^{\max,\tilde{S}}\in\mathcal{X}(\tilde{S}, \hat{a})$ such that $q\cdot\mathcal{E}_T\leq \zeta\sum_{i=1}^{N}|\mathcal{E}_T^i|\leq \hat{X}_T^{\max,\tilde{S}}$. It follows that $g+(\hat{X}_T^{\max,\tilde{S}}, \bar{\mathbf{0}})\in\mathbb{L}^0(\hat{K}_T)$. Moreover, by $(\ref{budccc})$, we have
\begin{equation}
\mathbb{E}[\langle g, Z_T\rangle]\leq x=\langle (x,\bar{\mathbf{0}}), Z_0\rangle.\nonumber
\end{equation}
Lemma $\ref{lemsp}$ implies the existence of an acceptable portfolio with $V_0=(x,\bar{\mathbf{0}})$ such that $V_T-g=V_T+(-\int_0^Tc_tdt+q\cdot\mathcal{E}_T,\bar{\mathbf{0}})\in\mathbb{L}^0(\hat{K}_T)$. Therefore, the conclusion holds that $c\in\mathcal{C}_{x,q\cdot\mathcal{E}_T}$.
\end{proof}

\begin{proof}[Proof of Lemma $\ref{equ1}$]
We first show that for all $(x,\bar{\mathbf{0}},q)\in\mathcal{K}$ and $z>0$,
\begin{equation}\label{sma}
x+\mathbb{E}[q\cdot\mathcal{E}_TZ_T^0]\geq z\mathbb{E}\Big[\int_0^Te^{\int_0^t(\delta_v-\alpha_v)dv}Z_t^0dt\Big],\ \forall Z\in\mathcal{Z}^s
\end{equation}
if and only if $\mathcal{A}(x,q,z)\neq\emptyset$.

On one hand, for the fixed $(x,\bar{\mathbf{0}},q)\in\mathcal{K}$ and $z>0$ such that $\mathcal{A}(x,q,z)\neq \emptyset$, there exists $c\in\mathbb{L}_+^0$ with $c_t\geq F(c)_t$ for $t\in[0,T]$ and
\begin{equation}
\mathbb{E}\Big[\int_0^Tc_tZ_t^0dt\Big]\leq x+\mathbb{E}[q\cdot\mathcal{E}_TZ_T^0],\ \ \forall Z\in\mathcal{Z}^s.\nonumber
\end{equation}
We claim that this choice $c_t\geq \bar{c}_t$ for $t\in[0,T]$ where $\bar{c}_t\equiv F(\bar{c})_t$ is the subsistent consumption process which equals its habit formation process all the time. To see this, by the definition of $F(c)_t$ and the constraint $c_t\geq F(c)_t$, it follows that
\begin{equation}
dF(c)_t\geq (\delta_tF(c)_t-\alpha_tF(c)_t)dt,\ \ F(c)_0=z.\nonumber
\end{equation}
Also, we always have
\begin{equation}
d\bar{c}_t=(\delta_t\bar{c}_t-\alpha_t\bar{c}_t)dt,\ \ \bar{c}_0=z,\nonumber
\end{equation}
from which, it can be solved that $\bar{c}_t=ze^{\int_0^t(\delta_v-\alpha_v)dv}$ for $t\in[0,T]$.

By subtraction, it follows that
\begin{equation}
e^{\int_0^t(\delta_v-\alpha_v)dv}(F(c)_t-\bar{c}_t)\geq 0,\ \ t\in[0,T].\nonumber
\end{equation}

It is therefore a consequence that $c_t\geq \bar{c}_t=ze^{\int_0^t(\delta_v-\alpha_v)dv}$ for $t\in[0,T]$, and by Lemma $\ref{budccc}$, we arrive at $(\ref{sma})$.

Consider $(x,\bar{\mathbf{0}},q)\in\mathcal{K}$ and $z>0$ such that $(\ref{sma})$ holds. We can always construct $c_t\equiv\bar{c}_t=ze^{\int_0^t(\delta_v-\alpha_v)dv}$ such that $c_t\equiv F(c)_t$ holds. The definition of $\mathcal{A}(x,q,z)$ and Lemma $\ref{budccc}$ yield that $c\in\mathcal{A}(x,q,z)$ and hence $\mathcal{A}(x,q,z)\neq\emptyset$.

So far, it has been proved that
\begin{equation}\label{part1}
\begin{split}
 \text{int}&\Big\{(x,q,z)\in\mathbb{R}^{N+2}: (x,\bar{\mathbf{0}},q)\in\mathcal{K},\ z>0\ \text{such that}\ \mathcal{A}(x,q,z)\neq \emptyset\Big\}\\
=&\Big\{ (x,q,z)\in\mathbb{R}^{N+2}: z>0\ \text{and}\  x+\mathbb{E}[q\cdot\mathcal{E}_TZ_T^0]>z\mathbb{E}\Big[\int_0^Te^{\int_0^t(\delta_v-\alpha_v)dv}Z_t^0dt\Big],\ \forall Z\in\mathcal{Z}^s\Big\}.
\end{split}
\end{equation}

It is enough to show the equivalence $\text{cl}\mathcal{B}_1=\mathcal{B}_2=\mathcal{B}_3$ for $z=0$, where we define
\begin{equation}
\begin{split}
&\mathcal{B}_1\triangleq \Big\{(x,q)\in\mathbb{R}^{N+1}:(x, \bar{\mathbf{0}}, q)\in\mathcal{K}\Big\},\\
&\mathcal{B}_2\triangleq \Big\{(x,q)\in\mathbb{R}^{N+1}:\mathcal{H}((x,\bar{\mathbf{0}}), q)\neq \emptyset\Big\},\\
&\mathcal{B}_3\triangleq \Big\{(x,q)\in\mathbb{R}^{N+1}: x+\mathbb{E}[q\cdot\mathcal{E}_TZ_T^0]\geq 0,\ \forall Z\in\mathcal{Z}^s\Big\}.\nonumber
\end{split}
\end{equation}

For the first equality, it is trivial to see that $\mathcal{B}_2\subseteq\text{cl}\mathcal{B}_1$. It is enough to verify that $\text{cl}\mathcal{B}_1\subseteq \mathcal{B}_2$. Choose $(x,q)\in\text{cl}\mathcal{B}_1$ and let $(x^n,q^n)_{n\geq 1}$ be a sequence in $\mathcal{B}_1$ that converges to $(x,q)$. The assertion of the lemma will follow if we show that $\mathcal{H}((x,\bar{\mathbf{0}}),q)\neq \emptyset$. Fix $V^n\in\mathcal{H}((x^n, \bar{\mathbf{0}}),q^n)$, $n\geq 1$, by Lemma $\ref{budlem}$, we have
\begin{equation}
\mathbb{E}[\langle V_T^n+(q^n\cdot\mathcal{E}_T, \bar{\mathbf{0}}),  Z_T\rangle]\leq x^n+\mathbb{E}[\langle (q^n\cdot\mathcal{E}_T,\bar{\mathbf{0}}), Z_T \rangle],\ \ \forall Z\in\mathcal{Z}^s.\nonumber
\end{equation}
Since $(x^n)_{n\geq 1}$ converges to $x$ and $(q^n)_{n\geq 1}$ converges to $q$, there exist constants $M_1>0$ and $M_2>0$ such that for $n$ large enough, we have $x^n<M_1$ and $q^n<M_2$. It follows that when $n$ is large, we have $V^n\in\mathcal{V}_{M_1}^{\text{acpt}}$ and $q_n\cdot\mathcal{E}_T\leq M_2\sum_{i=1}^{N}|\mathcal{E}_T^i|$. Lemma $\ref{ran}$ asserts that there exits a constant $\hat{a}>0$ and for each $\tilde{S}$, there exists an $\hat{X}^{\max,\tilde{S}}\in\mathcal{X}(\tilde{S}, \hat{a})$ such that $V_T^n+(\hat{X}_T^{\max,\tilde{S}}, \bar{\mathbf{0}})\in\mathbb{L}^0(\hat{K}_T)$ for $n$ large enough. We can then apply Lemma $\ref{bdp}$ and Proposition $3.4$ of \cite{Campi06}, passing if necessary to convex combinations, and assume that $V^n$ converges to a finite variation, predictable process $V$ pointwise. In particular, we obtain that $V_T+(q\cdot\mathcal{E}_T, \bar{\mathbf{0}})\in\mathbb{L}^0(\hat{K}_T)$ where $q\triangleq\lim_{n\rightarrow\infty} q^n$. Since $q\cdot\mathcal{E}_T\leq \zeta \sum_{i=1}^{N}|\mathcal{E}_T^i|$ where $\zeta=\max_{1\leq i\leq N}|q^i|$, by Lemma $\ref{ran}$ again, there exists a constant $a>0$ and for each $\tilde{S}\in\mathcal{S}^s$, there exists an $X^{\max,\tilde{S}}\in\mathcal{X}(\tilde{S}, a)$ such that $V_T+(X^{\max,\tilde{S}}_T, \bar{\mathbf{0}})\in\mathbb{L}^0(\hat{K}_T)$. Moreover, Fatou's lemma gives that
\begin{equation}
\begin{split}
\mathbb{E}[\langle V_T+(q\cdot\mathcal{E}_T, \bar{\mathbf{0}}), Z_T\rangle]&\leq \lim_{n\rightarrow\infty}\mathbb{E}[\langle V_T^n+(q^n\cdot\mathcal{E}_T, \bar{\mathbf{0}}), Z_T\rangle]\\
&\leq \lim_{n\rightarrow\infty} \Big(x^n+\mathbb{E}[\langle (q^n\cdot\mathcal{E}_T, \bar{\mathbf{0}}), Z_T\rangle]\Big)=x+\mathbb{E}[\langle (q\cdot\mathcal{E}_T, \bar{\mathbf{0}}), Z_T\rangle],\nonumber
\end{split}
\end{equation}
for all $Z\in\mathcal{Z}^s$. It follows that
\begin{equation}
\mathbb{E}[\langle V_T, Z_T\rangle]\leq x=\langle (x,\bar{\mathbf{0}}), Z_0\rangle ,\ \ \forall Z\in\mathcal{Z}^s.\nonumber
\end{equation}
By Lemma $\ref{lemsp}$, there exits an acceptable portfolio $\hat{V}$ with $\hat{V}_0=(x,\bar{\mathbf{0}})$ and $\hat{V}_T-V_T\in\mathbb{L}^0(\hat{K}_T)$, hence $\hat{V}\in \mathcal{H}((x,\bar{\mathbf{0}}), q)$.

To show $\mathcal{B}_2\subseteq\mathcal{B}_3$, for the fixed $(x,q)\in\mathcal{B}_2$, there exists a $V\in \mathcal{H}((x,\bar{\mathbf{0}}), q)$ and then Lemma $\ref{budlem}$ leads to
\begin{equation}
0\leq \mathbb{E}[\langle V_T+q\cdot\mathcal{E}_T, Z_T\rangle]\leq x+\mathbb{E}[\langle (q\cdot\mathcal{E}_T,\bar{\mathbf{0}}), Z_T\rangle],\ \ \forall Z\in\mathcal{Z}^s,\nonumber
\end{equation}
which completes the proof.

If $(x,q)\in\mathcal{B}_3$, define the $\mathcal{F}_T$-random variable $g\triangleq -q\cdot\mathcal{E}_T\geq -\zeta \sum_{i=1}^{N}|\mathcal{E}_T^i|$, where $\zeta=\max_{1\leq i\leq N}|q^i|$. Under Assumption $\ref{endowassum}$, by Lemma $\ref{ran}$, there exists a constant $a>0$ and for each $\tilde{S}$, there exits an $X^{\max,\tilde{S}}$ such that $(g+X^{\max,\tilde{S}},\bar{\mathbf{0}})\in\mathbb{L}^0(\hat{K}_T)$. Moreover, we have
\begin{equation}
\mathbb{E}[gZ_T^0]=\mathbb{E}[-q\cdot\mathcal{E}_TZ_T^0]\leq x=\langle (x,\bar{\mathbf{0}}), Z_0\rangle,\ \ \forall Z\in\mathcal{Z}^s.\nonumber
\end{equation}
Lemma $\ref{lemsp}$ guarantees the existence of an acceptable portfolio $\hat{V}$ with $\hat{V}_0=(x,\bar{\mathbf{0}})$ and $\hat{V}_T-(g,\bar{\mathbf{0}})\in\mathbb{L}^0(\hat{K}_T)$. Therefore, we have $\hat{V}_T+(q\cdot\mathcal{E}_T,\bar{\mathbf{0}})\in\mathbb{L}^0(\hat{K}_T)$ and $V\in\mathcal{H}((x,\bar{\mathbf{0}}, q))$, which verifies that $\mathcal{B}_3\subseteq\mathcal{B}_2$.

\end{proof}

\begin{proof}[Proof of Lemma $\ref{equivA}$]
Set $\tilde{c}_t=c_t-F(c)_t$ and it follows that
\begin{equation}
c_t=ze^{\int_0^t(\delta_v-\alpha_v)dv}+\tilde{c}_t+\int_0^t\delta_se^{\int_s^t(\delta_v-\alpha_v)dv}\tilde{c}_sds.\nonumber
\end{equation}
Denote $w_t\triangleq e^{\int_0^t(\delta_v\alpha_v)dv}$. By Fubini-Tonelli's theorem, we can deduce that
\begin{equation}\label{sim}
\begin{split}
\mathbb{E}\Big[\int_0^Tc_tZ_t^0dt\Big]&=z\mathbb{E}\Big[\int_0^Tw_tZ_t^0dt\Big]+\mathbb{E}\Big[\int_0^T\Big(\tilde{c}_t+\int_0^t\delta_se^{\int_s^t(\delta_v-\alpha_v)dv}\tilde{c}_sds\Big)Z_t^0dt\Big]\\
&=z\mathbb{E}\Big[\int_0^Tw_tZ_t^0dt\Big]+\mathbb{E}\Big[\int_0^T\tilde{c}_tZ^0_tdt+\int_0^T\delta_s\tilde{c}_s\Big(\int_s^Te^{\int_s^t(\delta_v-\alpha_v)dv}Z_t^0dt\Big)ds\Big]\\
&=z\mathbb{E}\Big[\int_0^Tw_tZ_t^0dt\Big]+\mathbb{E}\Big[\int_0^T\tilde{c}_tZ^0_tdt+\int_0^T\delta_t\tilde{c}_t\mathbb{E}\Big[\int_t^Te^{\int_t^s(\delta_v-\alpha_v)dv}Z_s^0ds\Big|\mathcal{F}_t\Big]dt\Big]\\
&=z\mathbb{E}\Big[\int_0^Tw_tZ_t^0dt\Big]+\mathbb{E}\Big[\int_0^T\tilde{c}_t\Gamma_tdt\Big].
\end{split}
\end{equation}

Following the similar computations in $(\ref{sim})$, we also obtain that
\begin{equation}
\mathbb{E}\Big[\int_0^Tw_tZ_t^0dt\Big]=\mathbb{E}\Big[\int_0^T\tilde{w}_t\Gamma_tdt\Big].\nonumber
\end{equation}

At last, it is observed that $Z_T^0=\Gamma_T$ by the definition of $\Gamma_T$, and hence $\mathbb{E}[q\cdot\mathcal{E}_TZ_T^0]=\mathbb{E}[q\cdot\mathcal{E}_T\Gamma_T]$, which completes the proof.
\end{proof}

\begin{proof}[Proof of Lemma $\ref{equ2}$]
It is enough to show that for any $(x,q,z)\in\mathbb{R}^{N+2}$, we have $\mathcal{A}(x,q,z)\neq \emptyset$ if and only if
\begin{equation}\label{assfo}
x+\mathbb{E}[q\cdot\mathcal{E}_TZ_T^0]-z\mathbb{E}\Big[\int_0^Te^{\int_0^t(\delta_v-\alpha_v)dv}Z_t^0dt\Big]\geq 0
\end{equation}
holds for all $Z\in\mathcal{Z}^s$.

If $\mathcal{A}(x,q,z)\neq\emptyset$, by its definition, there exists a $\hat{c}\in\mathbb{L}_+^0$ such that for any $Z\in\mathcal{Z}^s$ and hence any $\Gamma\in\widetilde{\mathcal{M}}$,
\begin{equation}
\begin{split}
0\leq \mathbb{E}\Big[\int_0^T\tilde{c}_t\Gamma_tdt\Big]&\leq x-z\mathbb{E}\Big[\int_0^T\tilde{w}_t\Gamma_tdt\Big]+\mathbb{E}[q\cdot\mathcal{E}_T\Gamma_T]\\
&=x-z\mathbb{E}\Big[\int_0^Te^{\int_0^t(\delta_v-\alpha_v)dv}Z_t^0dt\Big]+\mathbb{E}[q\cdot\mathcal{E}_TZ_T^0],\nonumber
\end{split}
\end{equation}
therefore $(\ref{assfo})$ holds trivially. On the other hand, if $(\ref{assfo})$ holds, it is enough to choose $c_t\equiv 0\in\mathcal{A}(x,q,z)$, which completes the proof.
\end{proof}

\subsection{Proofs of Main Results in Section $\ref{section3}$}
The proof of the Theorem \ref{mainth} is split into several results below.

For a vector $p=(p^0,p^1,\ldots,p^N)\in\mathbb{R}^{N+1}$, we denote $\mathcal{Z}^s(p)$ the subset of $\mathcal{Z}^s$ such that for $Z\in\mathcal{Z}^s(p)$,
\begin{equation}
\mathbb{E}\Big[\int_0^Te^{\int_0^t(\delta_v-\alpha_v)dv}Z_t^0dt\Big]=p^0,\ \ \ \mathbb{E}[\mathcal{E}_T^iZ_T^0]=p^i,\ \ 1\leq i\leq N.\nonumber
\end{equation}
Define the set $\mathcal{P}$ as the intersection of $\mathcal{L}$ with the hyperplane $y\equiv 1$. Given $p\in\mathcal{P}$, we can also define the auxiliary set
\begin{equation}
\widetilde{\mathcal{M}}(p)\triangleq \Big\{\Gamma\in\mathbb{L}_+^0:\Gamma_t=Z_t^0+\delta_t\mathbb{E}\Big[\int_t^Te^{\int_t^s(\delta_v-\alpha_v)dv}Z_s^0ds\Big|\mathcal{F}_t\Big],\ \forall t\in[0,T],\ Z\in\mathcal{Z}^s(p)\Big\}.\nonumber
\end{equation}
It follows that $\mathbb{E}\Big[\int_0^T\tilde{w}_t\Gamma_tdt\Big]=p^0$ and $\mathbb{E}[\mathcal{E}_T^i\Gamma_T]=p^i$, $1\leq i\leq N$ for all $\Gamma\in\widetilde{\mathcal{M}}(p)$.

\begin{lemma}\label{lem6}
Under all assumptions of Theorem $\ref{mainth}$, the set $\widetilde{\mathcal{M}}(p)$ is not empty if and only if $p\in\mathcal{P}$. In particular,
\begin{equation}
\bigcup_{p\in\mathcal{P}}\widetilde{\mathcal{M}}(p)=\widetilde{\mathcal{M}}.\nonumber
\end{equation}
\end{lemma}
\begin{proof}
Define the set $\mathcal{P}'\triangleq \{p\in\mathbb{R}^{N+1}: \widetilde{\mathcal{M}}(p)\neq\emptyset\}$. It is sufficient to verify that $\mathcal{P}=\mathcal{P}'$.

Following the proof of Lemma $8$ of \cite{kram04}, it is easy to show the direction $\mathcal{P}\subseteq\mathcal{P}'$.

For the other direction, let $p\in\mathcal{P}'$, $(x,q,z)\in\text{cl}\mathcal{L}$ and $\Gamma\in\widetilde{\mathcal{M}}(p)$. We claim the existence of $\tilde{c}\in\widetilde{\mathcal{A}}(x,q,z)$ such that $\bar{\mathbb{P}}[\tilde{c}>0]>0$. It then follows that
\begin{equation}
0<\mathbb{E}\Big[\int_0^T\tilde{c}_t\Gamma_tdt\Big]\leq x+(-z,q)\cdot p.\nonumber
\end{equation}
Since $(x,q,z)$ is chosen arbitrarily from $\text{cl}\mathcal{L}$, we obtain that $p\in\mathcal{P}$.

We now proceed to show that the claim holds. Choose any $(x,q,z)\in\text{cl}\mathcal{L}$ and denote the random variable $\Phi\triangleq -z\int_0^Te^{\int_0^t(\delta_v-\alpha_v)dv}dt+q\cdot\mathcal{E}_T$. Lemma $\ref{equ2}$ leads to
\begin{equation}
 x+\mathbb{E}[\langle (\Phi, \bar{\mathbf{0}}),Z_T\rangle ]\geq 0,\ \ \forall Z\in\mathcal{Z}^s,\nonumber
\end{equation}
which yields that
\begin{equation}
x+\inf_{Z\in\mathcal{Z}^s}\mathbb{E}[\langle (\Phi, \bar{\mathbf{0}}),Z_T\rangle ]\geq 0.\nonumber
\end{equation}
By the definition of $\widetilde{\mathcal{A}}(x,q,z)$ and the proof of Lemma $\ref{equivA}$, if all elements $\tilde{c}\in\widetilde{\mathcal{A}}(x,q,z)$ satisfy $\tilde{c}\equiv 0$,  we can deduce the existence of one $\hat{Z}\in\mathcal{Z}^s$ such that
\begin{equation}
x+\mathbb{E}[\langle (\Phi, \bar{\mathbf{0}}), \hat{Z}_T\rangle ]=0.\nonumber
\end{equation}
It implies that $\mathbb{E}[\langle (\Phi, \bar{\mathbf{0}}), \hat{Z}_T\rangle ]=\inf_{Z\in\mathcal{Z}^s}\mathbb{E}[\langle (\Phi, \bar{\mathbf{0}}),Z_T\rangle ]$. However, following the proof of Theorem $2.11$ of \cite{MR2113724} and part $(i)$ of Assumption $\ref{asssss}$, we can derive that for any $Z\in\mathcal{Z}^s$
\begin{equation}
\inf_{Z\in\mathcal{Z}^s}\mathbb{E}[\langle (\Phi, \bar{\mathbf{0}}),Z_T\rangle ]<\mathbb{E}[\langle (\Phi, \bar{\mathbf{0}}),Z_T\rangle ]<\sup_{Z\in\mathcal{Z}^s}\mathbb{E}[\langle (\Phi, \bar{\mathbf{0}}),Z_T\rangle ],\nonumber
\end{equation}
which is a contradiction.

\end{proof}

\begin{lemma}\label{lem7}
Let $p\in\mathcal{P}$, we have $\widetilde{\mathcal{M}}(p)\subseteq \widetilde{\mathcal{Y}}(1,p)$.
\end{lemma}
\begin{proof}
The result follows directly by Lemma $\ref{bcccc}$ and the definition of $\widetilde{\mathcal{M}}(p)$.
\end{proof}

\begin{lemma}\label{lem8}
Under all assumptions of Theorem $\ref{mainth}$, for any $(x,q,z)\in\mathcal{L}$, a nonnegative random variable $\tilde{c}\in\mathbb{L}_+^0(\Omega\times[0,T])$ belongs to $\widetilde{\mathcal{A}}(x,q,z)$ if and only if
\begin{equation}\label{ineqg}
\mathbb{E}\Big[\int_0^T \tilde{c}_t\Gamma_tdt\Big]  \leq x+(-z,q)\cdot p,\ \ \ \forall p\in\mathcal{P}\ \text{and}\ \Gamma\in\widetilde{\mathcal{M}}(p).
\end{equation}
\end{lemma}
\begin{proof}
If $\tilde{c}\in\widetilde{\mathcal{A}}(x,q,z)$, the inequality $(\ref{ineqg})$ follows directly by Lemma $\ref{bcccc}$ and Lemma $\ref{lem6}$. On the other hand, for any $\tilde{c}\in\mathbb{L}_+^0(\Omega\times[0,T])$ such that $(\ref{ineqg})$ holds, we have
\begin{equation}
\begin{split}
&\sup_{\Gamma\in\widetilde{\mathcal{M}}}\mathbb{E}\Big[\int_0^T\tilde{c}_t\Gamma_tdt+z\int_0^Te^{\int_0^t(-\alpha_v)dv}\Gamma_tdt-q\cdot\mathcal{E}_T\Gamma_T\Big]\\
=&\sup_{p\in\mathcal{P}}\sup_{\Gamma\in\widetilde{\mathcal{M}}(p)}\mathbb{E}\Big[\int_0^T\tilde{c}_t\Gamma_tdt+z\int_0^Te^{\int_0^t(-\alpha_v)dv}\Gamma_tdt-q\cdot\mathcal{E}_T\Gamma_T\Big]\\
=&\sup_{p\in\mathcal{P}}\sup_{\Gamma\in\widetilde{\mathcal{M}}(p)}\Big(\mathbb{E}\Big[\int_0^T\tilde{c}_t\Gamma_tdt\Big]+(z,-q)\cdot p\Big)\leq x .\nonumber
\end{split}
\end{equation}
It is a consequence of the definition of $\widetilde{\mathcal{A}}(x,q,z)$ that $\tilde{c}\in\widetilde{\mathcal{A}}(x,q,z)$.
\end{proof}

\begin{proposition}\label{lpro}
Under all assumptions of Theorem $\ref{mainth}$, the families $(\widetilde{\mathcal{A}}(x,q,z))_{(x,q,z)\in\mathcal{L}}$ and $(\widetilde{\mathcal{Y}}(y,r))_{(y,r)\in\mathcal{R}}$ have the following properties:
\begin{itemize}
\item[(1)] For any $(x,q,z)\in\mathcal{L}$, the set $\widetilde{\mathcal{A}}(x,q,z)$ contains a strictly positive random variable in $\mathbb{L}_+^0(\Omega\times[0,T])$. A nonnegative random variable $\tilde{c}\in\mathbb{L}_+^0(\Omega\times[0,T])$ belongs to $\widetilde{\mathcal{A}}(x,q,z)$ if and only if
\begin{equation}\label{ineqpro1}
\mathbb{E}\Big[\int_0^T\tilde{c}_t\Gamma_tdt\Big]\leq xy+(-z,q)\cdot r,\ \ \forall (y,r)\in\mathcal{R}\ \text{and}\ \Gamma\in\widetilde{\mathcal{Y}}(y,r).
\end{equation}
\item[(2)] For any $(y,r)\in\mathcal{R}$, the set $\widetilde{\mathcal{Y}}(y,r)$ contains a strictly positive random variable in $\mathbb{L}_+^0(\Omega\times[0,T])$. A nonnegative random variable $\Gamma\in\mathbb{L}_+^0(\Omega\times[0,T])$ belongs to $\widetilde{\mathcal{Y}}(y,r)$ if and only if
\begin{equation}\label{ineqpro2}
\mathbb{E}\Big[\int_0^T\tilde{c}_t\Gamma_tdt\Big]\leq xy+(-z,q)\cdot r,\ \ \forall (x,q,z)\in\mathcal{L}\ \text{and}\ \tilde{c}\in\widetilde{\mathcal{A}}(x,q,z).
\end{equation}
\end{itemize}
\end{proposition}
\begin{proof}
We first show that assertion $(1)$ holds. To this end, we choose $(x,q,z)\in\mathcal{L}$. Since $\mathcal{L}$ is an open set, there is a constant $\lambda>0$ such that $(x-\lambda, q,z)\in\mathcal{L}$. Let $\tilde{c}\in\widetilde{\mathcal{A}}(x-\lambda,q,z)$, since $\tilde{w}_t=e^{-\int_0^t\alpha_vdv}>0$ for $t\in[0,T]$, for any $\Gamma\in\widetilde{\mathcal{M}}$, we obtain that
\begin{equation}
\mathbb{E}\Big[\int_0^T\tilde{c}_t\Gamma_tdt\Big]\leq x-\lambda-z\mathbb{E}\Big[\int_0^T\tilde{w}_t\Gamma_tdt\Big]+\mathbb{E}[q\cdot\mathcal{E}_T\Gamma_T].\nonumber
\end{equation}

Under Assumption $\ref{asssss}$, let the constant $\beta\triangleq \sup_{Z\in\mathcal{Z}^s}\mathbb{E}\Big[\int_{0}^{T}e^{\int_0^t(\delta_v-\alpha_v)dv}Z_t^0dt\Big]<\infty$ and define the process $\rho_t\triangleq \frac{\lambda}{\beta}\tilde{w}_t>0$ for all $t\in[0,T]$. For all $\Gamma\in\widetilde{\mathcal{M}}$, it follows that
\begin{equation}
\begin{split}
\mathbb{E}\Big[\int_0^T\rho_t\Gamma_tdt\Big]\leq \mathbb{E}\Big[\int_0^T(\tilde{c}_t+\rho_t)\Gamma_tdt\Big]&\leq x-\lambda-z\mathbb{E}\Big[\int_0^T\tilde{w}_t\Gamma_tdt\Big]+\mathbb{E}[q\cdot\mathcal{E}_T\Gamma_T]+\frac{\lambda}{\beta}\mathbb{E}\Big[\int_0^T\tilde{w}_t\Gamma_tdt\Big]\\
&\leq x-\lambda-z\mathbb{E}\Big[\int_0^T\tilde{w}_t\Gamma_tdt\Big]+\mathbb{E}[q\cdot\mathcal{E}_T\Gamma_T]+\lambda\\
&\leq x-z\mathbb{E}\Big[\int_0^T\tilde{w}_t\Gamma_tdt\Big]+\mathbb{E}[q\cdot\mathcal{E}_T\Gamma_T].\nonumber
\end{split}
\end{equation}
The existence of a strictly positive random variable $\rho\in\widetilde{\mathcal{A}}(x,q,z)$ is a consequence of the definition of $\widetilde{\mathcal{A}}(x,q,z)$.

Assume that $(\ref{ineqpro1})$ holds for some $\tilde{c}\in\mathbb{L}_+^0$. By Lemma $\ref{lem7}$, we have $\Gamma\in\widetilde{\mathcal{M}}(p)\subset\widetilde{\mathcal{Y}}(1,p)$ for all $p\in\mathcal{P}$. Therefore, $(\ref{ineqg})$ holds, and by Lemma $\ref{lem8}$, we obtain that $\tilde{c}\in\widetilde{\mathcal{A}}(x,q,z)$. Conversely, let $\tilde{c}\in\widetilde{\mathcal{A}}(x,q,z)$, the inequality $(\ref{ineqpro1})$ follows by the definition of $\widetilde{\mathcal{Y}}(y,r)$, $(y,r)\in\mathcal{R}$.

For the proof of the assertion $(ii)$, as $k\widetilde{\mathcal{Y}}(y,r)=\widetilde{\mathcal{Y}}(ky,kr)$ for all $k>0$ and $(y,r)\in\mathcal{R}$, it is sufficient to consider the case $(y,r)=(1,p)$ for some $p\in\mathcal{P}$. The existence of a strictly positive $Z\in\mathcal{Z}^s(p)$ implies the existence of a strictly positive $\Gamma\in\widetilde{\mathcal{M}}(p)$. Lemma $\ref{lem7}$ again implies that $\Gamma\in\widetilde{\mathcal{Y}}(1,p)$ for $p\in\mathcal{P}$.

The second part follows directly from the definition of $\widetilde{\mathcal{Y}}(y,r)$.
\end{proof}

\begin{proof}[Proof of Theorem $\ref{mainth}$]\ \\
Once we get the abstract bipolar results in Proposition $\ref{lpro}$, under Assumptions $\ref{assumZ}$, $\ref{endowassum}$, $\ref{notrep}$, $\ref{ASSUV}$ and $\ref{asssss}$, the proof closely follows the arguments of Theorem $4.1$ and $4.2$ of \cite{Yu1}.
\end{proof}

\begin{proof}[Proof of Corollary $\ref{cormaincor}$]
If the process
\begin{equation}
Y^{\ast}_t(y,r)=\Gamma^{\ast}_t(y,r)-\delta_t\mathbb{E}\Big[\int_t^T\Gamma^{\ast}_s(y,r)e^{\int_t^s(-\alpha_v)dv} ds\Big|\mathcal{F}_t\Big]\nonumber
\end{equation}
is a strictly positive martingale, by using Fubini Theorem and tower property, it is easy to verify that
\begin{equation}
\Gamma_t^{\ast}(y,r)=Y^{\ast}_t(y,r)+\delta_t\mathbb{E}\Big[\int_t^T e^{\int_t^s(\delta_v-\alpha_v)dv} Y_s^{\ast}(y,r)ds\Big|\mathcal{F}_t\Big],\ \ \forall t\in[0,T].\nonumber
\end{equation}
Moreover, since both $\delta$ and $\alpha$ are constants, by changing the order of the conditional expectation and the integral, we get
\begin{equation}
\Gamma_t^{\ast}(y,r)=Y_t^{\ast}(y,r)\Big(1+\delta\int_t^Te^{(\delta-\alpha)(s-t)}ds\Big)=\left\{
\begin{array}{rl}
Y_t^{\ast}(y,r)\Big(\frac{\delta}{\delta-\alpha}e^{(\delta-\alpha)(T-t)}-\frac{\alpha}{\delta-\alpha}\Big),& \delta\neq\alpha,\\
Y_t^{\ast}(y,r)\Big(1+\delta(T-t)\Big),& \delta=\alpha.
\end{array}\right.\nonumber
\end{equation}
For the Logarithmic Utility function $U(t, x)=\log x$, we have $I(t,x)=\frac{1}{x}$ and hence part $(v)$ of Theorem $\ref{mainth}$ implies that the optimal consumption strategy is given explicitly by $(\ref{exlogconsum})$ and the corresponding optimal habit formation process is given explicitly by $(\ref{exlogconhab})$. Part $(iv)$ of Theorem $\ref{mainth}$ also implies that
\begin{equation}
\mathbb{E}\Big[\int_0^T\tilde{c}_t^{\ast}(x,q,z)\Gamma_t^{\ast}(y,r)dt\Big]=\mathbb{E}\Big[\int_0^T1dt\Big]=T=xy+(-z,q)\cdot r.\nonumber
\end{equation}

In addition, the assertion $(i)$ of Corollary $\ref{cormaincor}$ is an immediate consequence of the explicit formula $(\ref{exlogconhab})$. For the assertion $(ii)$, if $\alpha=0$, we get that
\begin{equation}
c_t^{\ast}(x,q,z)=ze^{\delta t}+\frac{1}{Y_t^{\ast}(y,r)e^{\delta(T-t)}}+\delta e^{\delta(t-T)} \int_0^t\frac{1}{Y_s^{\ast}(y,r)}ds.\nonumber
\end{equation}
Since $Y^{\ast}(y,r)$ is strictly positive, it is clear that if the discounting factor $\delta$ or the time horizon $T$ is sufficiently large, the second term and third term on the right hand side will be sufficiently small and $ze^{\delta t}$ will be the leading term. It thereby follows that $c_t^{\ast}(x,q,z)\approx ze^{\delta t}$ and it is an increasing process in terms of time $t$. Equivalently, we can conclude that the optimal consumption process satisfies the ratcheting constraint.

For assertion $(iii)$, if $\delta-\alpha\geq 0$, it is clear that
\begin{equation}
c_t^{\ast}(x,q,z)Y_t^{\ast}(y,r)=\left\{
\begin{array}{rl}
ze^{(\delta-\alpha)t}Y_t^{\ast}(y,r)+\frac{\delta-\alpha}{(\delta e^{(\delta-\alpha)(T-t)}-\alpha)}+\int_0^t\frac{(\delta-\alpha)Y_t^{\ast}(y,r)}{Y_s^{\ast}(y,r)}\frac{\delta e^{(\delta-\alpha)(t-s)}}{(\delta e^{(\delta-\alpha)(T-s)}-\alpha)}ds,& \delta>\alpha,\\
zY_t^{\ast}(y,r)+\frac{1}{1+\delta(T-t)}+\frac{1}{Y_t^{\ast}(y,r)}\int_0^t\frac{1}{Y_s^{\ast}(y,r)}\frac{\delta}{1+\delta(T-s)}ds,& \delta=\alpha.\nonumber
\end{array}\right.
\end{equation}
As $Y_t^{\ast}(y,r)$ is a martingale, it is easy to conclude that the product $c_t^{\ast}(x,q,z)Y_t^{y,r}$ is a submartingale. If we have $Y_0^{\ast}(y,r)=1$, by defining the equivalent probability measure $\frac{d\mathbb{Q}^{\ast}}{d\mathbb{P}}=Y_T^{\ast}=\Gamma_T^{\ast}(y,r)$, the optimal consumption process $c^{\ast}(x,q,z)$ is a submartingale under the measure $\mathbb{Q}^{\ast}$. In the special case that $\delta=\alpha=0$, i.e., the habit formation process keeps the constant initial habit $z$, the process $c^{\ast}(x,q,z)Y^{\ast}(y,r)$ is a martingale. Again, if $Y_0^{\ast}(y,r)=1$, the optimal consumption process is a martingale under the probability measure $\mathbb{Q}^{\ast}$. Moreover, by the explicit formula $(\ref{exlogconsum})$, if we have $Y_0^{\ast}(y,r)=1$, it is clear that the initial consumption amount is given by $(\ref{initialconm})$. Therefore, assertion $(iv)$ holds.
\end{proof}

\subsection{Proofs of Main Results in Section $\ref{section4}$}
\begin{proof}[Proof of Proposition $\ref{isotransc}$]
If the process $\delta_t-\alpha_t$ is a deterministic function of time $t$, the definition of the auxiliary dual set $\widetilde{\mathcal{M}}$ given in $(\ref{initialdualM})$ can be significantly simplified as
\begin{equation}
\widetilde{\mathcal{M}}=\Big\{\Gamma\in\mathbb{L}_+^0:\Gamma_t=Z_t^0G_t,\ \forall t\in[0,T],\ \ Z\in\mathcal{Z}^s\Big\},\nonumber
\end{equation}
where the process $(G_t)_{t\in[0,T]}$ is defined in $(\ref{ohG})$.

Equivalently, Lemma $\ref{equivA}$ can be rewritten in terms of SCPS $Z\in\mathcal{Z}^s$ by
\begin{equation}
\bar{\mathcal{A}}(x,q,z)=\Big\{\tilde{c}\in\mathbb{L}_+^0:\mathbb{E}\Big[\int_0^T\tilde{c}_tG_tZ_t^0dt\Big]\leq x-z\mathbb{E}\Big[\int_0^T\tilde{w}_tG_tZ_t^0dt\Big]+\mathbb{E}[q\cdot\mathcal{E}_TG_TZ_T^0],\ \forall Z\in\mathcal{Z}^s\Big\}.\nonumber
\end{equation}
Recall that $\mathcal{Z}^s$ is the set of all SCPS for the underlying asset $(S_t)_{t\in[0,T]}$ with transaction costs $\Lambda$ and $\bar{\mathcal{A}}(x,q,z)$ is the set of all auxiliary processes $(\tilde{c}_t)_{t\in[0,T]}$ that we define the auxiliary utility maximization problem without habit formation. Proposition $\ref{bcccc}$, i.e., the consumption budget constraint characterization, implies that the process $(\tilde{c}_tG_t)_{t\in[0,T]}$ is the financeable consumption process in the isomorphic market with the same underlying asset $(S_t)_{t\in[0,T]}$ and transaction costs $\Lambda$ and with the initial wealth $x$ and random endowments $N_T=-z\int_0^T\tilde{w}_tG_tdt+q\cdot\mathcal{E}_TG_T$. By the definition of $(G_t)_{t\in[0,T]}$, we notice that $G_T=1$ and therefore, $N_T$ can be simplified as $N_T=-z\int_0^T\tilde{w}_tG_tdt+q\cdot\mathcal{E}_T$. Denote $\hat{\mathcal{A}}(x,q,z)$ the set of all isomorphic consumption processes $\hat{c}_t=\tilde{c}_tG_t$, $t\in[0,T]$, where $\tilde{c}\in\bar{\mathcal{A}}(x,q,z)$, the auxiliary time-separable utility maximization problem becomes equivalently to
\begin{equation}
\hat{u}(x,q,z)=\sup_{\hat{c}\in\hat{\mathcal{A}}(x,q,z)}\mathbb{E}\Big[\int_0^TU(t,\frac{\hat{c}_t}{G_t})dt\Big],\nonumber
\end{equation}
where the external process $(G_t)_{t\in[0,T]}$ can be regarded as a discounting process or a num\'{e}raire process and the proof is complete.
\end{proof}

\begin{proof}[Proof of Corollary $\ref{exampleisomo1}$]
Let us consider the Logarithmic Utility function $U(t,x)=\log x$. Clearly, the utility maximization problem $(\ref{numeraireG})$ on $\hat{c}\in\hat{\mathcal{A}}(x,q,z)$ is equivalent to
\begin{equation}
\hat{u}(x,q,z)=\sup_{\hat{c}\in\hat{\mathcal{A}}(x,q,z)}\mathbb{E}\Big[\int_0^T\log(\hat{c}_t)dt\Big]-\mathbb{E}\Big[\int_0^T\log(G_t)dt\Big].\nonumber
\end{equation}
It is further equivalent to the standard utility maximization problem on the consumption in the isomorphic market model with random endowments $N_T$
\begin{equation}
\pi(x,q,z)=\sup_{\hat{c}\in\hat{\mathcal{A}}(x,q,z)}\mathbb{E}\Big[\int_0^T\log(\hat{c}_t)dt\Big].\nonumber
\end{equation}

\end{proof}

\begin{proof}[Proof of Corollary $\ref{exampleisomo2}$]
Suppose the external process $(G_t)_{t\in[0,T]}$ is a martingale with $G_0=1$,  Lemma $\ref{equivA}$ leads to
\begin{equation}
\begin{split}
\bar{\mathcal{A}}(x,q,z)=&\Big\{\tilde{c}\in\mathbb{L}_+^0:\mathbb{E}\Big[\int_0^T\tilde{c}_tG_tZ_t^0dt\Big]\leq x-z\mathbb{E}\Big[\int_0^T\tilde{w}_tG_tZ_t^0dt\Big]+\mathbb{E}[q\cdot\mathcal{E}_TG_TZ_T^0],\ \forall Z\in\mathcal{Z}^s\Big\}\\
=&\Big\{\tilde{c}\in\mathbb{L}_+^0:\mathbb{E}\Big[\int_0^T\tilde{c}_tZ_t^0dt G_T\Big]\leq x-z\mathbb{E}\Big[\int_0^T\tilde{w}_tZ_t^0dtG_T\Big]+\mathbb{E}[q\cdot\mathcal{E}_TZ_T^0 G_T],\ \forall Z\in\mathcal{Z}^s\Big\}\\
=&\Big\{\tilde{c}\in\mathbb{L}_+^0:\mathbb{E}^{\hat{\mathbb{P}}}\Big[\int_0^T\tilde{c}_tZ_t^0dt\Big]\leq x-z\mathbb{E}^{\hat{\mathbb{P}}}\Big[\int_0^T\tilde{w}_tZ_t^0dt\Big]+\mathbb{E}^{\hat{\mathbb{P}}}[q\cdot\mathcal{E}_TZ_T^0],\ \forall Z\in\mathcal{Z}^s\Big\},\nonumber
\nonumber
\end{split}
\end{equation}
where $\frac{d\hat{\mathbb{P}}}{\mathbb{P}}=G_T$. However, we know from the definition that $G_T=1$ and therefore, $\hat{\mathbb{P}}=\mathbb{P}$ which gives the simplified characterization of all auxiliary processes
\begin{equation}\label{ineqrealconsum}
\bar{\mathcal{A}}(x,q,z)=\Big\{\tilde{c}\in\mathbb{L}_+^0:\mathbb{E}\Big[\int_0^T\tilde{c}_tZ_t^0dt\Big]\leq x-z\mathbb{E}\Big[\int_0^T\tilde{w}_tZ_t^0dt\Big]+\mathbb{E}[q\cdot\mathcal{E}_TZ_T^0],\ \forall Z\in\mathcal{Z}^s\Big\}.
\end{equation}

It is worth noting that the num\'{e}raire process $(G_t)_{t\in[0,T]}$ disappears in the definition of the auxiliary set $\bar{\mathcal{A}}(x,q,z)$. Let us consider the isomorphic market model with stock price $(S_t)_{t\in[0,T]}$ and transaction costs $\Lambda$ and random endowments $R_T=-z\int_0^T\tilde{w}_tdt+q\cdot\mathcal{E}_T$. Proposition $\ref{bcccc}$ implies that each $(\tilde{c}_t)_{t\in[0,T]}$ is the $(x,R_T)$-financeable consumption process. Therefore, the auxiliary problem becomes a standard utility maximization on consumption without habit formation and we have the equivalence
\begin{equation}
u(x,q,z)=\sup_{\tilde{c}\in\bar{\mathcal{A}}(x,q,z)}\mathbb{E}\Big[\int_0^TU(t,\tilde{c}_t)dt\Big],\nonumber
\end{equation}
where $\bar{\mathcal{A}}(x,q,z)$ is the set of all $(x,R_T)$-financeable consumption processes in the market with the same asset process $(S_t)_{t\in[0,T]}$ and transaction costs $\Lambda$, which completes the proof.

\end{proof}
\ \\
\ \\
\textbf{Acknowledgement}:  The author thanks the associate editor and anonymous referees for their careful reviews and helpful suggestions concerning the presentation of this paper.
\ \\
\ \\
\bibliographystyle{plain}
\bibliography{ReferenceEndowment}

\end{document}